\documentclass{scrartcl}
\usepackage[numbers]{natbib}
\usepackage{fontenc}
\usepackage{shadethm}
\usepackage{amsthm}
\usepackage{float}
\usepackage{bm}
\usepackage{graphicx}
\usepackage{subfig}
\usepackage{stmaryrd}
\usepackage{mathrsfs} 
\usepackage{amsmath}
\usepackage{amssymb}
\usepackage{wasysym}
\usepackage{sectsty}
\usepackage[hyperfootnotes=false]{hyperref}
\usepackage[a4paper, left=2.7cm, right=2.7cm, top=2.5cm]{geometry}
\usepackage{chngcntr}
\counterwithin*{equation}{section}

\usepackage{cleveref}
\DeclareMathAlphabet{\mathpzc}{OT1}{pzc}{m}{it}

\sectionfont{\normalsize\centering}
\subsectionfont{\footnotesize\centering}

\theoremstyle{plain}
\newtheorem{thm}{Theorem}[section] 

\theoremstyle{definition}
\newtheorem{defn}[thm]{Definition} 
\newtheorem{lem}[thm]{Lemma}
\newtheorem{prop}[thm]{Proposition}
\newtheorem{rem}[thm]{Remark}
\newtheorem{cor}[thm]{Corollary}

\def\Xint#1{\mathchoice
	{\XXint\displaystyle\textstyle{#1}}%
	{\XXint\textstyle\scriptstyle{#1}}%
	{\XXint\scriptstyle\scriptscriptstyle{#1}}%
	{\XXint\scriptscriptstyle\scriptscriptstyle{#1}}%
	\!\int}
\def\XXint#1#2#3{{\setbox0=\hbox{$#1{#2#3}{\int}$ }
		\vcenter{\hbox{$#2#3$ }}\kern-.6\wd0}}

\def\dashint{\Xint-}
\usepackage[utf8]{inputenc}
\usepackage[german,english,russian]{babel}

\newcounter{MPequ}

\newcounter{AppA}

\begin{document}\selectlanguage{english}
\begin{center}
\normalsize \textbf{\textsf{Properties of the Biot-Savart operator acting on surface currents}}
\end{center}
\begin{center}
	Wadim Gerner\footnote{\textit{E-mail address:} \href{mailto:wadim.gerner@inria.fr}{wadim.gerner@inria.fr}}
\end{center}
\begin{center}
{\footnotesize	Sorbonne Universit\'e, Inria, CNRS, Laboratoire Jacques-Louis Lions (LJLL), Paris, France}
\end{center}
{\small \textbf{Abstract:} 
We investigate properties of the image and kernel of the Biot-Savart operator in the context of stellarator designs for plasma fusion. We first show that for any given coil winding surface (CWS) the image of the Biot-Savart operator is $L^2$-dense in the space of square-integrable harmonic fields defined on a plasma domain surrounded by the CWS. Then we show that harmonic fields which are harmonic in a proper neighbourhood of the underlying plasma domain can in fact be approximated in any $C^k$-norm by elements of the image of the Biot-Savart operator.

In the second part of this work we establish an explicit isomorphism between the space of harmonic Neumann fields and the kernel of the Biot-Savart operator which in particular implies that the dimension of the kernel of the Biot-Savart operator coincides with the genus of the coil winding surface and hence turns out to be a homotopy invariant among regular domains in $3$-space.

Lastly, we provide an iterative scheme which we show converges weakly in $W^{-\frac{1}{2},2}$-topology to elements of the kernel of the Biot-Savart operator.
\section{Introduction}
	In the context of stellarator designs for plasma fusion one distinguishing feature, as opposed to the tokamak design, is a complex coil structure \cite{Xu16}. There are different ways to model the complex coil structure in the stellarator design. It can be modelled as a collection of finite 1-dimensional closed curves, by means of a closed surface known as coil winding surface (CWS) around which the coils wrap or by means of a region enclosing a finite volume\footnote[2]{L.-M. Imbert-Gerard, E.J. Paul, and A.M. Wright. An Introduction to stellarators: From magnetic fields to symmetries and optimization. [Chapter 13.3] arXiv:1908.05360, to appear as a book in 2024.}.
	
	In the present work we focus on the CWS model and investigate the properties of the corresponding Biot-Savart operator. Specifically, let $S\subset\mathbb{R}^3$ be a regular enough (for the purpose of this introduction we assume all quantities involved to be $C^\infty$-smooth) closed surface and let $j\in \mathcal{V}(S)$ be a divergence-free vector field on $S$, i.e. $j\parallel S$ and $\operatorname{div}_S(j)=0$. By means of the Biot-Savart law \cite[Chapter 5]{Jack21} the magnetic field throughout $3$-space induced by such a surface current is given by
	\begin{gather}
		\label{1E1}
	\operatorname{BS}_S(j)(x):=\frac{1}{4\pi}\int_Sj(y)\times \frac{x-y}{|x-y|^3}d\sigma(y).
	\end{gather}
	In the context of stellarator design we require the hot fusion plasma to be confined in some region $P\subset\mathbb{R}^3$ which is surrounded by and of positive distance to $S$. I.e. if $\Omega$ is the bounded region enclosed by $S$ ($\partial \Omega=S$) then we require $P\subset \Omega$ and $\operatorname{dist}(S,P)>0$. So given such a plasma confinement region $P$ (known as plasma domain) and a coil winding surface $S$ with the mentioned properties it is of importance to understand the magnetic fields within $P$ which can be generated by currents being supported on $S$. Mathematically this amounts to understanding the image of the Biot-Savart operator.
	
	Before we present our findings let us first point out that the image and kernel of the Biot-Savart operator have been studied in different contexts. For example it is of relevance in fluid dynamics and magnetohydrodynamics to understand the properties of the Biot-Savart operator given by
	\[
	\operatorname{BS}_{\Omega}(v)(x):=\frac{1}{4\pi}\int_{\Omega}v(y)\times \frac{x-y}{|x-y|^3}d^3y\text{, }x\in \Omega
	\]
	where $\Omega\subset\mathbb{R}^3$ is a smooth bounded domain and $v$ is a smooth vector field on $\Omega$. This operator is used in order to define the helicity of vector fields which is a fundamental conserved quantity and may be regarded as a measure of the topological stability of the underlying fluid flow c.f. \cite{W58}, \cite{M69},\cite{BF84},\cite[Chapter III]{AK21}. A thorough study of this operator has been conducted in \cite{CDG01}.
	
	On the other hand, in the context of magnetic tomography, we may consider currents which are supported inside a bounded volume $V$ and look at the magnetic field induced by such a volume current on some surface $S$ surrounding the volume $V$ at a positive distance, i.e. one may consider the operator
	\[
	\operatorname{BS}^S_V(v)(x):=\frac{1}{4\pi}\int_{V}v(y)\times \frac{x-y}{|x-y|^3}d^3y\text{, }x\in S.
	\]
	Understanding the kernel of this operator is important for the question of uniqueness of current reconstructions, see for instance \cite{PW09},\cite{HKP05} and references therein.
	
	In our situation the current is supported on a surface $S$ and we are interested in the magnetic field inside a (plasma) domain $P$. It is a fact of Maxwell's equations of magnetostatics that in our situation we have
	\[
	\operatorname{curl}\left(\operatorname{BS}_S(j)\right)=0\text{, }\operatorname{div}\left(\operatorname{BS}_S(j)\right)=0\text{ in }P.
	\]
	Naively one may conjecture that $\operatorname{BS}_S$ may generate any such harmonic vector field on the plasma domain (of particular interest are square integrable harmonic fields since they correspond to magnetic fields of finite magnetic energy). However, we will provide a simple procedure which allows to construct harmonic fields which are not contained in the image, see \Cref{3P1}. Nonetheless we will be able to establish a density result, i.e. even though the image of the Biot-Savart operator does not contain all harmonic fields, it is dense in this space. To formulate a more precise statement let us define
	\[
	L^2\mathcal{H}(P):=\left\{B_T\in W^{1,2}(P,\mathbb{R}^3)\mid \operatorname{curl}(B_T)=0\text{, }\operatorname{div}(B_T)=0\right\}^{\|\cdot\|_{L^2(P)}},
	\]
	where $\cdot^{\|\cdot\|_{L^2(P)}}$ denotes the completion with respect to the norm $\|\cdot\|_{L^2(P)}$. The $T$ stands for "target", since the vector fields $B_T$ are the potential target magnetic fields which we wish to approximate by our surface current magnetic fields. The conditions $\operatorname{curl}(B_T)=0$ and $\operatorname{div}(B_T)=0$ are understood in the weak sense. Let us point out that we equivalently can write $L^2\mathcal{H}(P)=\{B_T\in L^2(P,\mathbb{R}^3)\mid \operatorname{curl}(B_T)=0=\operatorname{div}(B_T)\}$. Loosely speaking we prove the following result regarding the density, see \Cref{3C9} for a more precise statement. We denote by $D^2$ the closed unit disc in $\mathbb{R}^2$ and by $\mathcal{V}_{\operatorname{div}=0}(S)$ the smooth divergence-free tangent vector fields on $S$.
	\begin{thm}[First main result (informal version)]
		\label{T11}
		Let $\Omega\subset\mathbb{R}^3$ be a bounded domain with $\overline{\Omega}\cong D^2\times S^1$ and let $D^2\times S^1\cong \overline{P}\subset \Omega$. Suppose that $P$ wraps in toroidal direction exactly once around $\Omega$, then the image of the operator
		\[
		\operatorname{BS}_{\partial\Omega}:\mathcal{V}_{\operatorname{div}=0}(\partial\Omega)\rightarrow L^2\mathcal{H}(P)
		\]
		as defined in (\ref{1E1}) is dense in $L^2\mathcal{H}(P)$ with respect to the $\|\cdot\|_{L^2(P)}$-norm.
	\end{thm}
	We note that in stellarator designs the coil structure and plasma domain are taken to be toroidal and the coils are constructed in such a way that they wrap once in toroidal direction around the plasma domain, see for instance \cite[Figure 1]{KAl17} for the schematic structure of the Wendelstein 7-X stellarator in Greifswald Germany.
	
	From a mathematical point of view the plasma domain need not wrap around the volume $\Omega$ and might for example be contained in some contractible region. We show that in this case the image of the Biot-Savart operator is no longer dense so that the density of the image of the Biot-Savart operator in fact depends on the way the plasma domain is embedded in $\Omega$.
	
	From the point of view of applications, if we identify a nice target field $B_T\in L^2\mathcal{H}(P)$ which has good confinement properties, an $L^2$-approximation is usually not sufficient. In fact, the gradient $\nabla B_T$ of $B_T$ influences the "drift" of the plasma particles \cite[Chapter 2]{M16} and so if we wish to reproduce the desired plasma behaviour under the influence of the current induced magnetic field we should aim for a $C^1$-approximation. Of course, the $C^1$-norm may blow up even under $L^2$-closedness which may lead to undesirable plasma behaviour and so it is vital to identify a class of vector fields which can in fact be approximated in $C^1$-norm.
	
	Exploiting elliptic regularity results we can obtain the following corollary from \Cref{T11}.
	\begin{cor}
		\label{C12}
		Let $\Omega\subset \mathbb{R}^3$ be a bounded domain with $\overline{\Omega}\cong D^2\times S^1$ and let $D^2\times S^1\cong \overline{P}\subset \Omega$. Suppose that $P$ wraps once in toroidal direction around $\Omega$. If $\overline{P}\subset U\subset \Omega$ is an open set and $B_T\in L^2\mathcal{H}(U)$, then for every $\epsilon>0$ there exists a current $j_{\epsilon}\in \mathcal{V}_{\operatorname{div}=0}(S)$ with
		\[
		\|\operatorname{BS}_{S}(j_{\epsilon})-B_T\|_{C^1(\overline{P})}\leq \epsilon
		\]
		where $S=\partial \Omega$.
	\end{cor}
	Note that the approximation is valid on the plasma domain $P$ but not necessarily on all of $U$. More generally, by means of a bootstrapping argument, we may approximate $B_T$ on $P$ in any $C^k$-norm for any fixed $k\in \mathbb{N}_0$. As mentioned before, \Cref{C12} provides a vast class of magnetic fields on the plasma domain which can be well-approximated by magnetic fields induced by surface currents.
	
	The $C^1$-approximation property is a nice feature in view of the functionality of stellarator designs. However, as we will show, in order to achieve a better and better approximation it is sometimes necessary to utilise stronger and stronger currents so that due to physical limitations in real world applications there is at times only a finite accuracy which one can achieve. We state this observation in the following proposition.
	\begin{prop}
		\label{P13}
		Let $\Omega,P\subset \mathbb{R}^3$ be bounded domains with $\overline{P}\subset \Omega$. If $B_T\in L^2\mathcal{H}(P)\setminus \operatorname{Im}(\operatorname{BS}_S)$ where $S:=\partial \Omega$, then for any sequence $(j_n)_n\subset \mathcal{V}_{\operatorname{div}=0}(S)$ with $\|\operatorname{BS}_S(j_n)-B_T\|_{L^2(P)}\rightarrow 0$ as $n\rightarrow\infty$ we have $\|j_n\|_{L^2(S)}\rightarrow \infty$. Consequently, for any such $B_T$ there exists a constant $c>0$ such that for every $j\in \mathcal{V}_{\operatorname{div}=0}(S)$ with $\|j\|_{L^2(S)}\leq 1$ we have $\|\operatorname{BS}_S(j)-B_T\|_{L^2(P)}\geq c$.
	\end{prop}
Let us emphasise that the methods used in order to prove this result do not yield any numerical value for the constant $c$. So while we know that there is only a limited accuracy that may be achieved, this accuracy, a priori, might or might not be good enough for practical purposes.

The last question which we address in the present work is the characterisation of the kernel of the Biot-Savart operator. We have mentioned earlier that whether the image of the Biot-Savart operator is dense in the harmonic fields on the plasma domain or not depends on the way the plasma domain is embedded into the volume $\Omega$ bounded by the coil winding surface $S$. When it comes to the kernel the shape of the plasma domain is entirely irrelevant. This is simply because $\operatorname{BS}_S(j)$ is analytic within $\Omega$ for any surface current $j$ and so if $\operatorname{BS}_S(j)$ vanishes on $P$, then by real analyticity it must vanish on all of $\Omega$. Hence, mathematically, we are interested in the following question
\newline
\newline
\textit{Given a smooth bounded domain $\Omega\subset \mathbb{R}^3$ with boundary $S=\partial\Omega$. Is there a characterisation of} 
	\[
	\operatorname{Ker}(\operatorname{BS}_S)=\{j\in \mathcal{V}_{\operatorname{div}=0}(S)\mid \operatorname{BS}_S(j)(x)=0\text{ for all }x\in \Omega\}?
	\]
The relevance of this space is twofold. On the one hand it is of importance in view of inverse problems. Say, we are given a magnetic field $B$ inside the plasma domain $P$ which we know is induced by a surface current on $S$. Can the current $j$ inducing $B$, at least in principle, be uniquely obtained from the knowledge of $B$? A positive answer to this question corresponds to the kernel being trivial. On the other hand, if the kernel is non-trivial, adding any element of the kernel to some surface current $j$ will induce the same magnetic field inside $P$. But from the point of view of applications some current distributions may be easier to realise physically than others. So having a trivial current is a good property from the point of view of inverse problems while having a large kernel is a good property from the point of view of real life stellarator design as it grants some flexibility.

Before we state our main result regarding the kernel let us recall some basic concepts. Given a bounded smooth domain $\Omega\subset \mathbb{R}^3$ we can consider a smooth curl-free vector field $A$ on $\Omega$. In general, when $\Omega$ is not simply-connected, the vector field $A$ will not be a gradient field. One can then wonder how many "essentially different" such non-gradient curl-free vector fields exist. More precisely, two curl-free fields are said to be essentially the same if they differ by a gradient field and essentially different otherwise. Identifying curl-free fields which are essentially the same gives rise to an equivalence relation which induces the first de Rham cohomology group $H^1_{\operatorname{dR}}(\Omega)$ of the domain $\Omega$.

Additionally, we may introduce the notion of harmonic Neumann fields $\mathcal{H}_N(\Omega)$ on $\Omega$ which is the following space of smooth vector fields $\mathcal{H}_N(\Omega):=\{\Gamma\in \mathcal{V}(\Omega)|\operatorname{curl}(\Gamma)=0=\operatorname{div}(\Gamma)\text{, }\Gamma\parallel \partial\Omega\}$.

It is well-known that $\mathcal{H}_N(\Omega)$ is finite dimensional and that the Hodge-isomorphism establishes an isomorphism between $H^1_{\operatorname{dR}}(\Omega)$ and $\mathcal{H}_N(\Omega)$ \cite[Theorem 2.2.7 \& Theorem 2.6.1]{S95}. Even more, it is well-known that $\dim\left(H^1_{\operatorname{dR}(\Omega)}\right)$ is a homotopy invariant between smooth manifolds \cite[Theorem 17.11]{L12} and in fact we have $\dim(H^1_{\operatorname{dR}}(\Omega))=g(\partial\Omega)$, \cite{CDG02}, where $g(\partial\Omega)$ denotes the genus of $\partial\Omega$ (which, in the case $\partial\Omega$ is disconnected, is defined as the sum of the genera of the connected components of $\partial\Omega$).

Regarding the kernel of the Biot-Savart operator we will establish an explicit isomorphism between $\mathcal{H}_N(\Omega)$ and $\operatorname{Ker}(\operatorname{BS}_S)$ ($S=\partial\Omega$). This will, by the aforementioned Hodge isomorphism, in particular provide an isomorphism between $\operatorname{Ker}(\operatorname{BS}_S)$ and $H^1_{\operatorname{dR}}(\Omega)$. It then follows from the properties of the de Rham cohomology group that the dimension of the kernel of the Biot-Savart operator is a homotopy invariant. Note that it is highly non-trivial that this is the case because the Biot-Savart operator depends on $S$ (as it defines its domain of integration) but also the allowed currents depend on $S$ (since they are required to be tangent to $S$ and divergence-free with respect to the metric induced on $S$). So a-priori the dimension of the kernel of the Biot-Savart operator should be expected to depend not only on the topology of $S$ but also on the way $S$ is embedded into $3$-space.

We now state the informal version of our second main theorem (recall that $g(S)$ denotes the genus of a surface $S$).
\begin{thm}[Second main result (informal version)]
	\label{T14}
	Let $\Omega\subset \mathbb{R}^3$ be a bounded smooth domain. Then $\dim(\operatorname{Ker}(\operatorname{BS}_{\partial\Omega}))=g(\partial\Omega)$. In particular, if $\Omega_1$, $\Omega_2$ are two bounded, smooth domains in $\mathbb{R}^3$ which are homotopic, then the kernels of the respective Biot-Savart operators are isomorphic.
\end{thm}

\begin{center}
	\textbf{Structure of the paper}
\end{center}

In section 2 we introduce the notation and preliminary definitions needed to give a precise statement of the results. In section 3 we give precise statements and proofs of \Cref{T11}, \Cref{C12} and \Cref{P13}. We also discuss some additional results regarding the density of the image of the Biot-Savart operator, c.f. \Cref{3P6}. In section 4 we discuss the relation of the density result \Cref{T11} and some optimisation problems appearing in the plasma physics literature. In particular, we show how one can obtain a sequence of surface currents whose induced magnetic fields approximate a given target field better and better, c.f. \Cref{4C3} and \Cref{4R4}. In Section 5 we give a precise statment and proof of \Cref{T14} including a regularity result, see \Cref{3T11}. In section 6 we provide a recursive procedure which allows one to approximate the elements of the kernel of the Biot-Savart operator in some appropriate topology, see \Cref{6T5}.

To make this work more self-contained we included some results, which are well-known in the smooth setting, in the Appendix which deal with the setting of non-smooth domains. Appendix A contains a regularity result regarding the harmonic Neumann fields, which will be crucial in establishing regularity of the elements of the kernel of the Biot-Savart operator. Appendix B is devoted to the $L^2$-Hodge decomposition on non-smooth domains. In Appendix C we recall the notion of tangent traces and discuss their connection to the Biot-Savart operator.
\section{Notation and definitions}
Throughout this paper we denote for a given $C^{k,\alpha}$-regular ($k\in \mathbb{N},\alpha\in (0,1]$) hypersurface $S\hookrightarrow\mathbb{R}^3$ by $\mathcal{V}(S)$ the maximally smooth vector fields on $S$ (in particular these vector fields are taken to be tangent to $S$). We denote by $\mathcal{V}_0(S)$ the subspace of $\mathcal{V}(S)$ of divergence-free vector fields with respect to the induced Riemannian metric on $S$, equivalently $v\in \mathcal{V}_0(S)\Leftrightarrow \int_S v\cdot \operatorname{grad}(\phi)d\sigma(x)=0$ for all $\phi\in C^1_c(\mathbb{R}^3)$ (the space of compactly supported $C^1$-functions) where $d\sigma(x)$ is the standard induced surface measure. In addition, we denote by $\nabla_S\phi(x):=\nabla \widetilde{\phi}(x)-(\mathcal{N}(x)\cdot \nabla \widetilde{\phi}(x))\mathcal{N}(x)$ the surface gradient of any $C^1$-function $\phi\in C^1(S)$ where $\widetilde{\phi}$ is an arbitrary $C^1(\mathbb{R}^3)$-extension of $\phi$ and $\mathcal{N}$ any fixed unit normal field on $S$.
Given some $1\leq p\leq \infty$, $k\in \mathbb{N}_0$ we denote by $W^{k,p}\mathcal{V}(S)$, $W^{k,p}\mathcal{V}_0(S)$ the completion of the corresponding spaces $\mathcal{V}(S)$, $\mathcal{V}_0(S)$ with respect to the standard Sobolev $W^{k,p}$-norm, where as usual by convention $W^{0,p}\equiv L^p$ denotes the $L^p$-norm and $H^1\equiv W^{1,2}$. We will also make use of the following space defined on some given bounded domain $\Omega\subset \mathbb{R}^3$ 
\begin{gather}
	\nonumber
H(\operatorname{curl},\Omega):=\{w\in L^2\mathcal{V}(\Omega)\mid \operatorname{curl}(w)\in L^2\mathcal{V}(\Omega)\}
\end{gather}
where $\operatorname{curl}(w)$ is understood to exist in the weak sense and be of class $L^2$. We equip this space with the inner product $\langle v,w\rangle_{L^2,\operatorname{curl}}:=\int_{\Omega}v(x)\cdot w(x)d^3x+\int_{\Omega}\operatorname{curl}(v)(x)\cdot \operatorname{curl}(w)(x)d^3x$ which turns it into a Hilbert space.

In addition, given a bounded $C^1$-domain $\Omega\subset\mathbb{R}^3$ we denote by 
\[
\mathcal{H}_N(\Omega):=\{\Gamma\in H^1\mathcal{V}(\Omega)|\operatorname{curl}(\Gamma)=0=\operatorname{div}(\Gamma)\text{, }\Gamma\parallel \partial\Omega\},
\]
where $\operatorname{curl}$, $\operatorname{div}$ are the standard curl and div on the $3$-d domain $\Omega$ and $H^1\mathcal{V}(\Omega)$ denotes the $H^1$-regular vector fields on $\Omega$. Further we introduce the space of $L^p$-harmonic fields
\[
\mathcal{H}^p(\Omega):=\{B\in W^{1,p}\mathcal{V}(\Omega)|\operatorname{curl}(B)=0=\operatorname{div}(B)\}\text{, }
L^p\mathcal{H}(\Omega):=L^p\mathcal{H}^p(\Omega),
\]
where $\operatorname{curl}$ and $\operatorname{div}$ are understood in the weak sense, $L^p\mathcal{H}^p(\Omega)$ denotes the $L^p$-closure of $\mathcal{H}^p(\Omega)$ and we highlight that we do not enforce any boundary conditions.

Finally, we define the operator of interest. Let $P,\Omega\subset \mathbb{R}^3$ be bounded $C^{1,1}$-domains with $\overline{P}\subset \Omega$ and set $S:=\partial \Omega$, then we define 
\begin{gather}
	\nonumber
\operatorname{BS}^P_S:L^2\mathcal{V}_0(S)\rightarrow \mathcal{V}(P)\text{, }v\mapsto \left(x\mapsto\frac{1}{4\pi}\int_Sv(y)\times \frac{x-y}{|x-y|^3}d\sigma(y)\right),
\\
\nonumber
\operatorname{BS}_S:L^2\mathcal{V}_0(S)\rightarrow \mathcal{V}(\Omega)\text{, }v\mapsto \left(x\mapsto\frac{1}{4\pi}\int_Sv(y)\times \frac{x-y}{|x-y|^3}d\sigma(y)\right),
\end{gather}
where $\mathcal{V}(\Omega)$ denotes the smooth vector fields on $\Omega$. Throughout we always assume that the domains involved are $C^{1,1}$-regular unless otherwise noted.
\section{(Non-)Density of the image}
\subsection{$\operatorname{Im}(\operatorname{BS}_S)\subsetneq L^p\mathcal{H}(P)$}
In this subsection we establish the following simple preliminary observation.
\begin{prop}
	\label{3P1}
	Let $P,\Omega\subset \mathbb{R}^3$ be bounded $C^{1,1}$-domains with $\overline{P}\subset \Omega$. Then
	\[
	\operatorname{Im}(\operatorname{BS}^P_S)\subsetneq L^p\mathcal{H}(P)
	\]
	for every $1\leq p\leq \infty$, where $S=\partial\Omega$.
\end{prop}
\begin{proof}[Proof of \Cref{3P1}]
	After translating $\Omega$ if necessary, we may assume that $0\in \Omega\setminus \overline{P}$. We then define the vector field $B_T(x):=\frac{x}{|x|^3}\in L^p\mathcal{H}(P)$ for every $1\leq p\leq \infty$ (since $0$ has a positive distance to $P$). We observe that $B_T$ is analytic on $\Omega\setminus \{0\}$. Further, observe that for every $j\in L^2\mathcal{V}_0(S)$ we have $\operatorname{BS}_S(j)\in \mathcal{V}(\Omega)$ and even more, $\operatorname{BS}_S(j)$ is div- and curl-free in $\Omega$ and thus it is analytic on $\Omega$. Consequently $\operatorname{Im}(\operatorname{BS}_S^P)\subset L^p\mathcal{H}(P)$ for all $1\leq p\leq \infty$ and if there were to exist some $j\in L^2\mathcal{V}_0(S)$ with $\operatorname{BS}_S(j)|_P=\operatorname{BS}_S^P(j)=B_T$ on $P$, then by real analyticity of the vector fields involved we would find $\operatorname{BS}_S(j)=B_T$ on $\Omega\setminus \{0\}$. However, $B_T$ blows up if we approach $0$, while $\operatorname{BS}_S(j)$ remains bounded. Thus, $B_T\in L^p\mathcal{H}(P)$ cannot be contained in the image of $\operatorname{BS}_S^P$.
\end{proof}
\begin{rem}
	Even if the domain of $\operatorname{BS}_S^P$ is replaced by $L^1\mathcal{V}(S)$ the same argument shows that the image of this operator will be strictly contained in the space $L^p\mathcal{H}(P)$ for every $1\leq p\leq \infty$.
\end{rem}
\subsection{Proof of \Cref{P13}}
Before we come to the proof of the first main result \Cref{T11} let us first turn to \Cref{P13} which, in comparison, is much easier to establish. Let us recall the statement
\begin{prop}
	\label{3P3}
	Let $P,\Omega\subset \mathbb{R}^3$ be bounded $C^{1,1}$-domains with $\overline{P}\subset \Omega$. If $B_T\in L^2\mathcal{H}(P)\setminus \operatorname{Im}(\operatorname{BS}^P_S)$ where $S:=\partial\Omega$, then for any sequence $(j_n)_n\subset L^2\mathcal{V}_0(S)$ with $\|\operatorname{BS}_S^P(j_n)-B_T\|_{L^2(P)}\rightarrow0$ as $n\rightarrow\infty$ we have $\|j_n\|_{L^2(S)}\rightarrow\infty$. Consequently, for any such $B_T$ there exists some $c>0$ such that for every $j\in L^2\mathcal{V}_0(S)$ with $\|j\|_{L^2(S)}\leq 1$ we have $\|\operatorname{BS}_S^P(j)-B_T\|_{L^2(P)}\geq c$.
\end{prop}
\begin{proof}[Proof of \Cref{3P3}]
	Assume that there exists a bounded sequence $(j_n)_n\subset L^2\mathcal{V}_0(S)$ such that $(\operatorname{BS}_S^P(j_n))_n$ converges to $B_T$ in $L^2(P)$. Then by reflexivity of Hilbert spaces, after possibly passing to a subsequence, we may assume that the $(j_n)_n$ converge weakly in $L^2(S)$ to some $j\in L^2\mathcal{V}_0(S)$. But one easily verifies (since $\operatorname{dist}(P,S)>0$) that the operator $\operatorname{BS}_S^P$ is $L^2(S)\text{-}L^2(P)$ continuous and so $\operatorname{BS}_S^P(j_n)$ must converge weakly to $\operatorname{BS}_S^P(j)$ in $L^2(P)$. By assumption $\operatorname{BS}_S^P(j_n)$ converges to $B_T$ in $L^2(P)$ and hence we must have $B_T=\operatorname{BS}_S^P(j)\in \operatorname{Im}(\operatorname{BS}_S^P)$ which is a contradiction.
\end{proof}
\subsection{Proof of \Cref{T11}}
The strategy of the proof of \Cref{T11} relies on the abstract functional analytic fact that if $(V,\|\cdot\|)$ is a normed vector space and $U\leq V$ a subspace, then $U$ is dense in $V$ if and only if the annihilator $U^\circ:=\{T\in V^\prime|T(u)=0\text{ for all }u\in U\}\leq V^\prime$ satisfies $U^\circ=\{0\}$, where $V^\prime$ denotes the topological dual space of $V$, \cite[Corollary 16.8 \& Remark 16.9]{Haase14}. We hence wish to understand the annihilator of $\operatorname{Im}(\operatorname{BS}_S^P)$. To this end the following characterisation of the topological dual space of $L^2\mathcal{H}(P)$ will be useful.
\begin{lem}
	\label{3L5}
	Let $P\subset \mathbb{R}^3$ be a bounded $C^{1,1}$-domain. Then the following map
	\[
	\mathcal{I}:L^2\mathcal{H}(P)\rightarrow \left(L^2\mathcal{H}(P)\right)^{\prime}\text{, }B\mapsto \left(L^2\mathcal{H}(P)\rightarrow\mathbb{R}\text{, }A\mapsto\int_PA\cdot Bd^3x\right)
	\]
	is a linear isomorphism.
\end{lem}
\begin{proof}[Proof of \Cref{3L5}]
	This is an immediate consequence of the Riesz representation theorem.
\end{proof}
\begin{rem}
	\label{3ExtraRem5}
	Using the duality of the standard $L^p$-spaces \cite[Chapter 4 Theorem 2.1]{HL99} one can show more generally, whenever a suitable $L^p$-Hodge decomposition is available, that the map $\mathcal{I}$ defines a linear isomorphism from $L^q\mathcal{H}(P)$ into $\left(L^p\mathcal{H}(P)\right)^\prime$ for every $1<p<\infty$ where $1<q<\infty$ denotes the corresponding H\"{o}lder conjugate. A suitable $L^p$-Hodge decomposition is for instance available if the underlying domain $P$ satisfies an additional cut-property, see \cite[Hypothesis 1.1, Theorem 6.1 \& Corollary 6.1]{AS12} or if $P$ is $C^\infty$-smooth \cite[Corollary 3.5.2]{S95}.
\end{rem}
In a first step we prove a non-density result whenever the plasma domain $P$ has disconnected boundary.
\begin{prop}
	\label{3P6}
	Let $P,\Omega\subset\mathbb{R}^3$ be bounded $C^{1,1}$-domains with $\overline{P}\subset \Omega$. If $\partial P$ is disconnected, then $\operatorname{Im}(\operatorname{BS}_S^P)$ is not $L^p(P)$-dense in $L^p\mathcal{H}(P)$ for any $1\leq p< \infty$, where $S:=\partial\Omega$.
\end{prop}
\begin{proof}[Proof of \Cref{3P6}]
	By assumption $\partial P$ has at least two boundary components which we denote by $C_1,\dots,C_N$ for some $2\leq N<\infty$ (and the labelling can be chosen in an arbitrary way). We then consider the following boundary value problem
	\[
	\Delta f=0\text{ in }P\text{, }f|_{C_j}=\delta_{1j}\text{ for }1\leq j\leq N,
	\]
	where $\delta_{ik}$ denotes the standard Kronecker delta. Because $\partial P\in C^{1,1}$ this problem has a solution $f\in \bigcap_{1\leq q<\infty}W^{2,q}(P)$, \cite[Theorem 2.4.2.5]{Gris85}. Since $\partial P$ is disconnected we see that $f$ is not constant and hence $B:=\operatorname{grad}(f)\in W^{1,q}\mathcal{V}(P)\setminus \{0\}$ for every $1\leq q<\infty$. Further, since $f|_{\partial\Omega}$ is locally constant, we find $B\perp \partial P$. We will now show that $B$ is in the annihilator of $\operatorname{Im}(\operatorname{BS}_S)\leq L^p\mathcal{H}(P)$ for every $1\leq p< \infty$. To this end we define the following two operators
	\begin{gather}
		\nonumber
		T:L^p\mathcal{H}(P)\rightarrow \mathbb{R}\text{, }A\mapsto \int_PA\cdot Bd^3x,
		\\
		\nonumber
		\operatorname{BS}_P(B)(x):=\frac{1}{4\pi}\int_PB(y)\times \frac{x-y}{|x-y|^3}d^3y\text{, }x\in \mathbb{R}^3.
	\end{gather}
	We note that by Sobolev embeddings \cite[Chapter 5.6.2 Theorem 5]{Evans10} we have $B\in L^\infty\mathcal{V}(P)$ and therefore $T\in \left(L^p\mathcal{H}(P)\right)^\prime$ for every $1\leq p< \infty$. We claim that $T\in \left(\operatorname{Im}(\operatorname{BS}^P_S)\right)^\circ$ for which we have to show that for any $A\in \operatorname{Im}(\operatorname{BS}_S^P)$ we have $T(A)=0$. Equivalently we have to show that
	\[
	\int_P\operatorname{BS}_S^P(j)(x)\cdot B(x)d^3x=0\text{ for all }j\in L^2\mathcal{V}_0(S).
	\]
	But we observe that it follows immediately by writing out the definition that $\int_P\operatorname{BS}^P_S(j)\cdot Bd^3x=\int_Sj\cdot \operatorname{BS}_P(B)d\sigma(x)$ and further for any fixed $x\in \Omega$ we have
	\begin{gather}
		\nonumber
		4\pi\operatorname{BS}_{P}(B)(x)=\int_{P}B(y)\times \frac{x-y}{|x-y|^3}d^3y=-\int_{P}\nabla_y\left(\frac{1}{|x-y|}\right)\times B(y)d^3y
		\\
		\nonumber
		=-\int_{P}\operatorname{curl}_y\left(\frac{B(y)}{|x-y|}\right)d^3y+\int_{P}\frac{\operatorname{curl}(B)(y)}{|x-y|}d^3y,
	\end{gather}
	where we used that $B\in W^{1,q}\mathcal{V}(P)$ for every $q>3$ and $\frac{1}{|x-\cdot|}\in W^{1,\alpha}(P)$ for all $1\leq \alpha<\frac{3}{2}$ so that $\frac{B(y)}{|x-y|}\in W^{1,1}\mathcal{V}(P)$ and the standard calculus identity $\operatorname{curl}\left(\psi B\right)=(\nabla \psi)\times B+\psi\operatorname{curl}(B)$ for scalar functions $\psi$ and vector valued functions $B$. In addition, letting $e_i$ denote the standard basis vectors and using an integration by parts, we find
	\begin{gather}
		\nonumber
		\int_{P}\operatorname{curl}_y\left(\frac{B}{|x-y|}\right)d^3y=\left\langle\operatorname{curl}\left(\frac{B}{|x-\cdot|}\right),e_i\right\rangle_{L^2(P)}e_i=\int_{\partial P}\frac{B(y)\times e_i}{|x-y|}\cdot \mathcal{N}(y)d\sigma(y)e_i
		\\
		\nonumber
		=\int_{\partial P}\left(\mathcal{N}(y)\times \frac{B(y)}{|x-y|}\right)\cdot e_id\sigma(y)e_i=\int_{\partial P}\mathcal{N}(y)\times \frac{B(y)}{|x-y|}d\sigma(y).
	\end{gather}
	We therefore arrive at
	\[
	4\pi\operatorname{BS}_{P}(B)(x)=-\int_{\partial P}\frac{\mathcal{N}(y)\times B(y)}{|x-y|}d\sigma(y)+\int_{P}\frac{\operatorname{curl}(B)(y)}{|x-y|}d^3y.
	\]
	Since $B\perp \partial P$ and $\operatorname{curl}(B)=0$, we conclude $\operatorname{BS}_P(B)=0$, see also \cite[Theorem B \& Theorem B$^\prime$]{CDG01} for the smooth setting, and consequently $0\neq T\in \left(\operatorname{Im}(\operatorname{BS}_S^P)\right)^\circ$ and thus the image of $\operatorname{BS}_S^P$ is not dense in $L^p\mathcal{H}(P)$ for any $1\leq p< \infty$.
\end{proof}
Now we recall that $\mathcal{H}_N(P)=\{\Gamma\in H^1\mathcal{V}(P)|\operatorname{curl}(\Gamma)=0=\operatorname{div}(\Gamma)\text{, }\Gamma\parallel \partial P\}$ and we note that if $P$ is a $C^{1,1}$-domain, then according to \Cref{ALemma1} we have $\mathcal{H}_N(P)\subset W^{1,p}\mathcal{V}(P)$ for all $1\leq p<\infty$ and thus $\mathcal{H}_N(P)\subset L^p\mathcal{H}(P)$ for every $1\leq p<\infty$.

The main result regarding the density of the Biot-Savart operator on $C^{1,1}$-domains is the following.
\begin{thm}
	\label{3T7}
	Let $P,\Omega\subset \mathbb{R}^3$ be bounded $C^{1,1}$-domains. Assume that $\overline{P}\subset \Omega$ and that $\partial P$ is connected. Then we have
	\begin{gather}
		\nonumber
		\mathcal{I}^{-1}\left(\left(\operatorname{Im}(\operatorname{BS}_S^P)\right)^\circ\right)\leq \mathcal{H}_N(P),
	\end{gather}
	where $\mathcal{I}$ denotes the isomorphism from \Cref{3L5}, $S:=\partial\Omega$ and the annihilator is considered with respect to the space $\left(L^2\mathcal{H}(P),\|\cdot\|_{L^2(P)}\right)$.
\end{thm}
\begin{rem}
	In accordance with \Cref{3ExtraRem5} it follows that if $P$ satisfies an additional cut-assumption or is $C^\infty$-smooth, then for every $1<p<\infty$ we have the same inclusion $\mathcal{I}^{-1}\left((\operatorname{Im}(\operatorname{BS}^P_S))^\circ\right)\leq \mathcal{H}_N(P)$ where the annihilator is considered with respect to the space $\left(L^p\mathcal{H}(P),\|\cdot\|_{L^p(P)}\right)$.
\end{rem}
\begin{proof}[Proof of \Cref{3T7}] To ease notation we let $S:=\partial\Omega$. We then have the following equivalence
\begin{gather}
	\nonumber
	B\in \mathcal{I}^{-1}\left((\operatorname{Im}(\operatorname{BS}^P_S))^\circ\right)\Leftrightarrow \mathcal{I}(B)\in (\operatorname{Im}(\operatorname{BS}^P_S))^\circ\Leftrightarrow \mathcal{I}(B)(A)=0\text{ for all }A\in \operatorname{Im}(\operatorname{BS}^P_S)
	\\
	\nonumber
	\Leftrightarrow \mathcal{I}(B)(\operatorname{BS^P_S}(j))=0\text{ for all }j\in L^2\mathcal{V}_0(S)\Leftrightarrow \int_PB\cdot \operatorname{BS}^P_S(j)d^3x=0\text{ for all }j\in L^2\mathcal{V}_0(S).
\end{gather}
We recall from the proof of \Cref{3P6} that we have the identity $\int_P B\cdot \operatorname{BS}_S^P(j)d^3x=\int_S \operatorname{BS}_P(B)\cdot jd\sigma(x)$ and we notice that since $\operatorname{dist}(P,S)>0$ the vector field $\operatorname{BS}_P(B)$ is smooth in a neighbourhood of $S$. We can now use the musical isomorphism to identify $\operatorname{BS}_P(B)$ with a $1$-form $\omega^1_{\operatorname{BS}_P(B)}$ and we observe that by density
\begin{gather}
	\nonumber
B\in \mathcal{I}^{-1}\left((\operatorname{Im}(\operatorname{BS}^P_S))^\circ\right)\Leftrightarrow\int_S j\cdot \operatorname{BS}_P(B)d\sigma(x)=0\text{ for all }j\in L^2\mathcal{V}_0(S)
\\
\label{3E3}
\Leftrightarrow \int_S j\cdot \operatorname{BS}_P(B)d\sigma(x)=0\text{ for all }j\in \mathcal{V}_0(S)\Leftrightarrow \iota^\#\omega^1_{\operatorname{BS}_P(B)}\text{ is an exact $1$-form,}
\end{gather}
where $\iota:S\rightarrow \mathbb{R}^3$ denotes the inclusion map and $\iota^\#$ denotes the corresponding pullback.

We claim that
\begin{equation}
	\label{3E4}
	\iota^{\#}\omega^1_{\operatorname{BS}_P(B)}\text{ is closed on }S\Leftrightarrow B\in \mathcal{H}_N(P)
\end{equation}
which will prove the theorem. One direction is immediate, namely if $B\in \mathcal{H}_N(P)$, then by extending $B$ by zero outside of $P$ one easily verifies that this extension is divergence-free in the weak sense on $\mathbb{R}^3$. It is then standard that the Biot-Savart potential is of class $H^1_{\operatorname{loc}}\mathcal{V}(\mathbb{R}^3)$ and satisfies $\operatorname{curl}\left(\operatorname{BS}_P(B)\right)=\operatorname{curl}\left(\operatorname{BS}_{\mathbb{R}^3}(\chi_PB)\right)=\chi_PB$, \cite[Corollary 5.3.15]{G20Diss}, where $\chi_P$ denotes the characteristic function of $P$. So in particular, since $S\subset \mathbb{R}^3\setminus \overline{P}$, $\operatorname{curl}(\operatorname{BS}_P(B))=0$ on $S$ in this case and hence in terms of differential forms this is the same as the identity $d\omega^1_{\operatorname{BS}_P(B)}=0$ on $\overline{P}^c$ and so the closedness follows by means of the fact that pullbacks commute with exterior differentiation.

Now we suppose that $\iota^{\#}\omega^1_{\operatorname{BS}_P(B)}$ is closed. By means of the Hodge-decomposition theorem, \Cref{BTheorem1}, we can then decompose $B$ as $B=\Gamma+\operatorname{grad}(f)$ for some $f\in W^{1,2}(P)$ and $\Gamma\in \mathcal{H}_N(P)$, where we used that $B$ is curl-free. As pointed out $\iota^\#\omega^1_{\operatorname{BS}_P(\Gamma)}$ is closed. Hence $\iota^\#\omega^1_{\operatorname{BS}_P(B)}$ is closed if and only if $\iota^\#\omega^1_{\operatorname{BS}_P(\operatorname{grad}(f))}$ is closed due to the linearity of the Biot-Savart operator. Further, we have the following equivalence from standard differential geometric facts
\begin{gather}
	\nonumber
d\iota^\#\omega^1_{\operatorname{BS}_P(B)}=0\Leftrightarrow \iota^\#\left(\star \star d\omega^1_{\operatorname{BS}_P(B)}\right)=0\Leftrightarrow \star n\left( \star d\omega^1_{\operatorname{BS}_P(B)}\right)=t\left(\star \star d\omega^1_{\operatorname{BS}_P(B)}\right)=0
\\
\nonumber
\Leftrightarrow n\left( \star d\omega^1_{\operatorname{BS}_P(B)}\right)=0 \Leftrightarrow \operatorname{curl}\left(\operatorname{BS}_P(B)\right)\parallel S,
\end{gather}
where $\star$ denotes the Hodge star operator and we used the notions of the tangent and normal part of a differential form and the "dual" relation between these notions \cite[Chapter 1.2 equation 2.25 \& proposition 1.2.6]{S95}. Hence
\[
d\iota^\#\omega^1_{\operatorname{BS}_P(\operatorname{grad}(f))}=0\Leftrightarrow \operatorname{curl}\left(\operatorname{BS}_P(\operatorname{grad}(f))\right)\parallel S.
\]
Letting $h(x):=-\int_P\operatorname{grad}(f)(y)\cdot \frac{x-y}{|x-y|^3}d^3y$, a direct calculation yields
\[
4\pi\operatorname{curl}\left(\operatorname{BS}_P(\operatorname{grad}(f))\right)=\operatorname{grad}(h)\text{ on }\mathbb{R}^3\setminus \overline{P}.
\]
In particular $\Delta h=0$ on $\mathbb{R}^3\setminus \overline{P}$, so that in fact $h$ is analytic on $\mathbb{R}^3\setminus \overline{P}$. We now set $E:=-\operatorname{grad}(h)$ and observe that $\operatorname{div}(E)=0$ on $\mathbb{R}^3\setminus \overline{P}$ so that $\operatorname{div}(h E)=-|E|^2$ and consequently for all $R\geq r$ for any fixed $r>0$ with $\overline{\Omega}\subset B_r(0)$, setting $U_r:=B_r(0)\setminus \overline{\Omega}$, we find
\begin{gather}
	\nonumber
	\|E\|^2_{L^2(U_r)}\leq\|E\|^2_{L^2(U_R)} =-\int_{U_R}\operatorname{div}(hE)d^3x
	\\
	\nonumber
	=-\int_{\partial B_R(0)}h E\cdot \frac{x}{|x|}d\sigma(x)+\int_{S}h E\cdot \mathcal{N}d\sigma(x)=-\int_{\partial B_R(0)}h E\cdot \frac{x}{|x|}d\sigma(x),
\end{gather}
where $\mathcal{N}$ denotes the outward unit normal with respect to $\Omega$ on $S=\partial\Omega$ and we used that $E=-\operatorname{grad}(h)=-4\pi\operatorname{curl}\left(\operatorname{BS}_P(\operatorname{grad}(f))\right)\parallel S$.
We finally observe that for $|x|\gg 1$ we have $|h(x)|\leq \frac{c}{|x|^2}$ and $|E(x)|\leq \frac{c}{|x|^3}$ for some suitable constant $c>0$ (which depends on $\operatorname{grad}(f)$ but is independent of $R$ or $x$). Consequently, taking the limit $R\rightarrow\infty$, we conclude that $E=0$ on $U_r$ for every $r>0$ with $\overline{\Omega}\subset B_r(0)$. Hence $\operatorname{grad}(h)=-E=0$ on $\mathbb{R}^3\setminus \overline{\Omega}$ and by real analyticity $\operatorname{grad}(h)=0$ on $\mathbb{R}^3\setminus \overline{P}$ because $\partial P$ is connected and hence so is $\mathbb{R}^3\setminus \overline{P}$, \cite{Li88}. In conclusion $h=0$ (due to its behaviour as $|x|\rightarrow \infty$) on $\mathbb{R}^3\setminus \overline{P}$. It then follows from \cite[Theorem 9.9]{GT01}, since $\operatorname{grad}(f)\in L^2\mathcal{V}(P)$, that $h\in W^{1,2}_{\operatorname{loc}}(\mathbb{R}^3)$ because $h(x)=\partial_{x^i}\int_P\frac{\partial_if(y)}{|x-y|}d^3y=-4\pi\partial_{x^i}N(\partial_if)(x)$ where $N(\phi)$ denotes the Newton potential of a function $\phi$. This in particular implies that $h\in W^{1,2}(P)$ and that the trace of $h$ when viewed as a function on $P$ coincides with its trace when viewed as a function on $\mathbb{R}^3\setminus \overline{P}$ \cite[Theorem 3.44]{DD12} so that, because $h=0$ on $\mathbb{R}^3\setminus \overline{P}$, we infer $h\in W^{1,2}_0(P)$. We will see later in \cref{3L8} that we can express
\begin{equation}
	\label{3E5}
	h(x)=\int_{\partial P}f(y)\frac{y-x}{|y-x|^3}\cdot \mathcal{N}(y)d\sigma(y)-4\pi f(x)\text{ for all }x\in P.
\end{equation}
A direct calculation yields then that for all $x\in P$ we have
\[
\Delta_x \int_{\partial P}f(y)\frac{y-x}{|y-x|^3}\cdot \mathcal{N}(y)d\sigma(y)=0. 
\]
In addition, since $\operatorname{div}(B)=0$ in the weak sense, we have $\Delta f=0$ in $P$ and consequently $\Delta h=0$ in $P$. So in particular $h\in W^{1,2}_0(P)$ satisfies $\Delta h=0$ in the weak sense and the uniqueness of weak solutions to the Dirichlet-Laplace problem implies $h=0$ in $P$. Overall
\[
h=0 \text{ on }\mathbb{R}^3\setminus \partial P.
\]
Now we approximate $\operatorname{grad}(f)$ in $L^2$-norm by div-free vector fields $(B_n)_n\subset \bigcap_{1\leq s<\infty}W^{1,s}\mathcal{V}(P)$. Then by means of the Hardy-Littlewood-Sobolev inequality, \cite[Chapter V]{S70}, $h_n(x):=-\int_PB_n(y)\cdot \frac{x-y}{|x-y|^3}d^3y$ will converge to $h(x)$ in $L^6(\mathbb{R}^3)$. We then have for any fixed $\phi\in C^\infty_c(\mathbb{R}^3)$
\[
0=\int_{\mathbb{R}^3}(\Delta \phi) hd^3x=\lim_{n\rightarrow\infty}\int_{\mathbb{R}^3}(\Delta \phi) h_nd^3x.
\]
Due to the regularity of the $B_n$ we have $h_n(x)=-\int_{\partial P}\frac{B_n(y)\cdot \mathcal{N}(y)}{|x-y|}d\sigma(y)$ and  $\operatorname{grad}(h_n)(x)=\int_{\partial P}(B_n\cdot \mathcal{N})\frac{x-y}{|x-y|^3}d\sigma(y)$. Consequently
\[
\int_{\mathbb{R}^3}(\Delta \phi)h_nd^3x=-\int_{\mathbb{R}^3}\operatorname{grad}(\phi)\cdot \operatorname{grad}(h_n)d^3x=4\pi \int_{\partial P}(B_n\cdot \mathcal{N})\phi(y)d\sigma(y)=4\pi \int_P\operatorname{grad}(\phi)\cdot B_nd^3y,
\]
where we used Fubini's theorem in combination with the explicit expression for $\operatorname{grad}(h_n)$ and the facts that $\frac{x-y}{|x-y|^3}=-\nabla_x\frac{1}{|x-y|}$, $-\Delta_x \frac{1}{|x-y|}=4\pi \delta(x-y)$ where $\delta(x)$ denotes the Dirac delta and we used Gauss' formula to establish the last identity (recall that the $B_n$ are divergence-free). We therefore arrive at
\[
0=\lim_{n\rightarrow\infty}\int_P\operatorname{grad}(\phi)\cdot B_nd^3y=\int_P\operatorname{grad}(\phi)\cdot \operatorname{grad}(f)d^3y
\]
by the approximation property. We recall that $\phi\in C^\infty_c(\mathbb{R}^3)$ was arbitrary so that by a density argument the above identity remains to hold for all $\phi\in W^{1,2}(P)$. Letting $\phi=f$ yields $\operatorname{grad}(f)=0$ and in conclusion $B=\Gamma\in \mathcal{H}_N(P)$.
\end{proof}
To finish the proof we are left with justifying (\ref{3E5}).
\begin{lem}
	\label{3L8}
	Let $P\subset \mathbb{R}^3$ be a bounded $C^1$-domain, $1<p<\infty$, and $f\in W^{1,p}(P)$ be a weak solution of $\Delta f=0$ in $P$. Then
	\[
	-\int_P\operatorname{grad}(f)(y)\cdot \frac{x-y}{|x-y|^3}d^3y=\int_{\partial P}f(y)\frac{y-x}{|y-x|^3}\cdot \mathcal{N}(y)d\sigma(y)-4\pi \chi_P(x)f(x)\text{ for all }x\in \mathbb{R}^3\setminus \partial P.
	\] 
\end{lem}
\begin{proof}[Proof of \Cref{3L8}] We prove the statement for $x\in P$, the other situation ($x\in \mathbb{R}^3\setminus \overline{P}$) can be dealt with in the same manner but without the need to cut out an $\epsilon$-ball. Since $x\in P$, we have $B_{\epsilon}(x)\subset P$ for all small enough $\epsilon>0$. In addition, since $\Delta f=0$ in $P$ we know that $f$ is in fact analytic in $P$. We can then write
\[
\int_P\operatorname{grad}(f)(y)\cdot \frac{x-y}{|x-y|^3}d^3y=\int_{P\setminus B_{\epsilon}(x)}\operatorname{grad}(f)(y)\cdot \frac{x-y}{|x-y|^3}d^3y+\int_{B_{\epsilon}(x)}\operatorname{grad}(f)(y)\cdot \frac{x-y}{|x-y|^3}d^3y.
\]
Since $B_{\epsilon}(x)\subset P$ and $f$ is analytic in $P$ we can estimate $|\operatorname{grad}(f)(y)|\leq c(x)$ for all $y\in B_{\epsilon}(x)$ independent of $y$ for all small enough $\epsilon$ and where $c(x)>0$ is a constant which will depend on the point $x$. With this it is clear that
\[
\lim_{\epsilon\searrow 0}\int_{B_{\epsilon}(x)}\operatorname{grad}(f)(y)\cdot \frac{x-y}{|x-y|^3}d^3y=0.
\]
On the other hand, using an integration by parts and that $\operatorname{div}_y\left(\frac{x-y}{|x-y|^3}\right)=0$ for all $y\in P\setminus B_{\epsilon}(x)$ we obtain
\[
\int_{P\setminus B_{\epsilon}(x)}\operatorname{grad}(f)(y)\cdot \frac{x-y}{|x-y|^3}d^3y=\int_{\partial P}f(y)\frac{x-y}{|x-y|^3}\cdot \mathcal{N}(y)d\sigma(y)+\int_{\partial B_{\epsilon}(x)}f(y)\frac{x-y}{|x-y|^3}\cdot \frac{x-y}{|x-y|}d\sigma (y).
\]
Since $f$ is continuous in $P$ we then find
\[
\lim_{\epsilon\searrow 0}\int_{\partial B_{\epsilon}(x)}f(y)\frac{x-y}{|x-y|^3}\cdot \frac{x-y}{|x-y|}d\sigma (y)=4\pi f(x)
\]
which altogether concludes the proof. \end{proof}
Regarding the proof that $E=0$ on $\mathbb{R}^3\setminus \overline{\Omega}$ in the proof of \Cref{3T7} we refer the reader also to \cite[Proof of Theorem A]{CDG01} where a similar reasoning was used.

\Cref{3T7} enables us to prove some (non-)density results which include \Cref{T11}. Before we state the result let us introduce same nomenclature. If $\Omega\subset \mathbb{R}^3$ is a $C^1$-domain with $\overline{\Omega}\cong D^2\times S^1$, i.e. $\Omega$ is a solid torus, then we call a closed $C^1$-curve $\gamma$ contained in $\partial\Omega$ a poloidal curve if, under identification of said diffeomorphism, it is homotopic to the standard closed curve sweeping out $\partial D^2\times \{x\}$ for some fixed point $x\in S^1$. In terms of fundamental groups, if we fix any $x$ in the image of a poloidal curve $\gamma$, then we can identify $\pi_1(\overline{\Omega},x)\cong \mathbb{Z}$ and $\pi_1(\partial \Omega,x)\cong \mathbb{Z}\times \mathbb{Z}$ (where we think of the first factor to correspond to windings around $\partial D^2$), and $\gamma$ will then represent the trivial element in $\pi_1(\overline{\Omega},x)$ while it represents the element $(1,0)$ in $\pi_1(\partial\Omega,x)$. Further, given some subset $D\subset \mathbb{R}^3$ we say that $D$ is a $C^2$-disc if $D$ is a $2$-dimensional $C^2$-submanifold with $C^1$-boundary which is diffeomorphic to the unit disc of $\mathbb{R}^2$, i.e. $\overline{D}$ itself is a $C^1$-embedded submanifold with boundary and its (manifold) interior $D$ is $C^2$-embedded.
\begin{cor}
	\label{3C9}
	Let $P,\Omega\subset \mathbb{R}^3$ be bounded $C^{1,1}$-domains with $\overline{P}\subset \Omega$ and set $S:=\partial\Omega$.
	\begin{enumerate}
		\item If $\partial \Omega\cong S^2$ and $P$ is any $C^{1,1}$-domain with connected boundary. Then $\operatorname{Im}(\operatorname{BS}_S^P)$ is $L^2(P)$-dense in $L^2\mathcal{H}(P)$ if and only if $\partial P$ is diffeomorphic to a sphere.
		\item If $\overline{P}\cong D^3$, where $D^3$ denotes the closed unit ball in $\mathbb{R}^3$, then $\operatorname{Im}(\operatorname{BS}_S^P)$ is $L^2(P)$-dense in $L^2\mathcal{H}(P)$.
		\item In the following we assume that $\overline{\Omega}\cong D^2\times S^1\cong \overline{P}$ are solid tori.
		\begin{enumerate}
			\item If $\partial \Omega$ contains a poloidal $C^1$-curve $\gamma$ which bounds a $C^2$-disc $D\subset \Omega$ with $D\cap \overline{P}=\emptyset$, then $\operatorname{Im}(\operatorname{BS}_S^P)$ is not $L^2(P)$-dense in $L^2\mathcal{H}(P)$.
			\item If $\partial\Omega$ contains a poloidal $C^1$-curve $\gamma_1$ which bounds a $C^2$-disc $D\subset \Omega$ such that $D\cap P$ is yet again a $C^2$-disc which is bounded by a poloidal $C^1$-curve $\gamma_2$ contained in $\partial P$, then $\operatorname{Im}(\operatorname{BS}_S^P)$ is $L^2(P)$-dense in $L^2\mathcal{H}(P)$.
		\end{enumerate}
	\end{enumerate}
\end{cor}
For a depiction of the two situations described in (iii) see the following \Cref{3F2}.
\begin{figure}[H]
	\centering
	\begin{subfloat}
		\centering
		\includegraphics[width=0.25\textwidth, keepaspectratio]{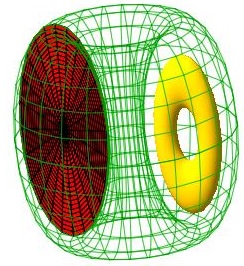}
	\end{subfloat}
	\hspace{2cm}
	\begin{subfloat}
		\centering
		\includegraphics[width=0.25\textwidth, keepaspectratio]{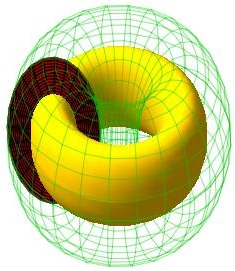}
	\end{subfloat}
	\caption{Left side case (a), right side case (b). The coil winding surface $S=\partial\Omega$ is depicted by the green grid. The plasma domain $P$ is depicted in yellow. The disc $D$ is depicted in red. In both cases the disc is bounded by a poloidal curve contained in $S$.}
	\label{3F2}
\end{figure}
\begin{proof}[Proof of \Cref{3C9}]
	\underline{(i):} If $\partial\Omega\cong S^2$, then any closed $1$-form on $\partial\Omega$ is exact. It then follows from (\ref{3E3}) and (\ref{3E4}) that the annihilator of $\operatorname{Im}(\operatorname{BS}_S^P)$ is zero if and only if $\mathcal{H}_N(P)=\{0\}$ which in turn is the case if and only if the genus of $\partial P$ is zero which is equivalent to $\partial P$ being a sphere.
	\newline
	\newline
	\underline{(ii):} If $\overline{P}\cong D^3$, then $\mathcal{H}_N(P)=\{0\}$ so that \Cref{3T7} tells us that the annihilator of $\operatorname{Im}(\operatorname{BS}_S^P)$ is zero and thus $\operatorname{Im}(\operatorname{BS}_S^P)$ is $L^2(P)$-dense in $L^2\mathcal{H}(P)$.
	\newline
	\newline
	\underline{(iii):} We recall that according to (\ref{3E3}) and (\ref{3E4}) it is enough to check whether for a fixed $B\in \mathcal{H}_N(P)\setminus \{0\}$ the closed $1$-form $\iota^\#\omega^1_{\operatorname{BS}_P(B)}$ is exact on $\partial \Omega$ or not. If it is exact for at least one such $B$, the image of $\operatorname{BS}_S^P$ will not be dense in $L^2\mathcal{H}(P)$ and if it is not exact for any such $B$, the annihilator will be zero and the density-property will follow.
	\newline
	\newline
	\underline{(a):} We observe first that a closed $1$-form on $\partial\Omega$ is exact if and only if it integrates to zero along a set of generators of the first fundamental group of $\partial \Omega$. We can first again identify $\pi_1(\partial \Omega,x)\cong \mathbb{Z}\times \mathbb{Z}$ for a fixed $x\in \partial \Omega$ where the first factor corresponds to windings around $\partial D^2$ (recall that $\partial \Omega\cong \partial D^2\times S^1$). It is then possible to find a representative $\widetilde{\gamma}$ of the element $(0,1)$ which bounds a disc $\widetilde{D}\cong D^2$ outside of $\Omega$. We recall that $\operatorname{BS}_P(B)$ is curl-free outside of $\overline{P}$ and thus $\omega^1_{\operatorname{BS}_P(B)}$ is closed on $\mathbb{R}^3\setminus \overline{P}$ and since $\widetilde{D}\subset \mathbb{R}^3\setminus \Omega \subset \mathbb{R}^3\setminus \overline{P}$ we conclude that $\iota_{\widetilde{D}}^\#\omega^1_{\operatorname{BS}_P(B)}$ is closed, where $\iota_{\widetilde{D}}:\widetilde{D}\rightarrow\mathbb{R}^3$ denotes the inclusion map. Then by Stokes' theorem we find
	\[
	\int_{\widetilde{\gamma}} \iota^\#\omega^1_{\operatorname{BS}_P(B)}=\int_{\widetilde{D}}d\iota_D^\#\omega^1_{\operatorname{BS}_P(B)}=0.
	\]
	On the other hand, by assumption, we can now find a poloidal curve $\gamma$ contained in $\partial\Omega$ which bounds a disc $D\cong D^2$ in $\Omega$ such that $D\cap \overline{P}=\emptyset$. Then we can argue as before that $\iota_{D}^\#\omega^1_{\operatorname{BS}_P(B)}$ is closed and thus Stokes' theorem yields $\int_{\gamma}\iota^\#\omega^1_{\operatorname{BS}_P(B)}=0$. In conclusion $\iota^\#\omega^1_{\operatorname{BS}_P(B)}$ is exact on $\partial\Omega$ and (\ref{3E3}) implies that the annihilator of $\operatorname{Im}(\operatorname{BS}_S^P)$ is non-zero, i.e. $\operatorname{Im}(\operatorname{BS}_S^P)$ is not dense in $L^2\mathcal{H}(P)$.
	\newline
	\newline
	\underline{(b):} By (\ref{3E3}) and (\ref{3E4}) it is enough to show that $\iota^\#\omega^1_{\operatorname{BS}_P(B)}$ is not exact on $\partial\Omega$. This will follow once we show that there exists a closed $C^1$-loop $\gamma$ on $\partial\Omega$ with $\int_{\gamma}\iota^\#\omega^1_{\operatorname{BS}_P(B)}\neq 0$. To this end we consider a poloidal curve $\gamma_1$ on $\partial\Omega$ which bounds a disc $D_1$ such that $D_1\cap \overline{P}$ is again a disc denoted by $D_2$ which is bounded by a poloidal curve $\gamma_2$ contained in $\partial P$. We recall that $\operatorname{BS}_P(B)$ is smooth away from $\partial P$ and satisfies $\operatorname{curl}(\operatorname{BS}_P(B))=\chi_PB$ where $\chi_P$ denotes the characteristic function of $P$. We can then consider a small tubular neighbourhood $U_{\epsilon}$ around $\partial P$ which has two boundary components $C^1_{\epsilon}\subset \overline{P}^c$ and $C^2_{\epsilon}\subset P$. We then have
	\[
	\int_{D_2\setminus U_{\epsilon}}B\cdot \mathcal{N}d\sigma(x)=\int_{D_2\setminus U_{\epsilon}}\operatorname{curl}(\operatorname{BS}_P(B))\cdot \mathcal{N}d\sigma(x)=\int_{D_1\setminus U_{\epsilon}}\operatorname{curl}(\operatorname{BS}_P(B))\cdot \mathcal{N}d\sigma(x)
	\]
	where we used that $\operatorname{curl}(\operatorname{BS}_P(B))=0$ outside of $P$. We observe first that $B$ is continuous up to the boundary by Sobolev embeddings since $\mathcal{H}_N(P)\subset W^{1,p}\mathcal{V}(P)$ for all $1\leq p<\infty$. Hence the left hand side, upon taking the limit $\epsilon\searrow 0$ converges to $\int_{D_2}B\cdot \mathcal{N}d\sigma(x)$ and so
	\[
	\int_{D_2}B\cdot \mathcal{N}d\sigma(x)=\lim_{\epsilon\searrow 0}\int_{D_1\setminus U_{\epsilon}}\operatorname{curl}(\operatorname{BS}_P(B))\cdot \mathcal{N}d\sigma(x).
	\]
	We may assume that $D_1\setminus U_{\epsilon}$ is bounded by $3$ closed curves, the poloidal curve $\gamma_1$ and two more curves $\gamma^3_{\epsilon}\subset C^1_{\epsilon}$ and $\gamma^4_{\epsilon}\subset C^2_{\epsilon}$. Then applying Stokes' theorem yields
	\[
	\int_{D_1\setminus U_{\epsilon}}\operatorname{curl}(\operatorname{BS}_P(B))\cdot \mathcal{N}d\sigma(x)=\int_{\gamma_1}\iota^\#\omega^1_{\operatorname{BS}_P(B)}+\int_{\gamma_3^\epsilon}\operatorname{BS}_P(B)+\int_{\gamma_4^\epsilon}\operatorname{BS}_P(B).
	\]
	As $\epsilon\searrow 0$ the curves $\gamma_3^\epsilon$, $\gamma_4^\epsilon$ will approach the curve $\gamma_2$ but due to the orientation conditions in Stokes' theorem the limiting curve of $\gamma_3^\epsilon$ and $\gamma_4^\epsilon$ will be oriented in the opposite way. Then due to the continuity of $\operatorname{BS}_P(B)$ these path integrals cancel as $\epsilon\searrow 0$ and therefore we obtain
	\[
	\int_{\gamma_1}\iota^\#\omega^1_{\operatorname{BS}_P(B)}=\int_{D_2}B\cdot \mathcal{N}d\sigma(x).
	\]
	To be more rigorous one may start with a $C^1$-vector field on $D_1$ which is everywhere tangent to $D_1$ and outward pointing along $\gamma_2$ and extend it to a $C^1$-vector field $X$ on $\Omega$ in such a way that $X$ is everywhere outward pointing on $\partial P$ and compactly supported within a small neighbourhood around $\partial P$. Letting $\Psi_t$ denote the global flow of $X$ we can then consider the tubular neighbourhood $U_{\epsilon}$ given by the flowout $\Psi:\partial P\times (-\epsilon,\epsilon)\rightarrow U_{\epsilon}$, $(x,t)\mapsto \Psi_t(x)$ which defines for small enough $0<\epsilon$ a $C^1$-diffeomorphism onto some open neighbourhood $U_{\epsilon}$ of $\partial P$. It is clear from construction that $D_1\setminus U_{\epsilon}$ is bounded by the curves $\gamma_1$, $\Psi_{\epsilon}\circ \gamma_2$ and $\Psi_{-\epsilon}\circ \gamma_2$. From here one can easily justify the above argument rigorously. 
	
	It then follows from \cite[Corollary 3.4]{AS12} and the fact that $D_2$ is bounded by a poloidal curve that $\int_{D_2}B\cdot \mathcal{N}d\sigma(x)= 0$ implies $B=0$ and hence the annihilator consists solely of the zero functional, i.e. $\operatorname{Im}(\operatorname{BS}_S^P)$ is $L^2(P)$-dense in $L^2\mathcal{H}(P)$. 
\end{proof}
\subsection{Proof of \Cref{C12}}
In this subsection we introduce a class of "nice" harmonic fields within which the image of the Biot-Savart operator is $C^k$-dense for every $k\in \mathbb{N}$, drastically improving the $L^2$-density result, \Cref{3C9}.

Essentially, we consider those harmonic fields, which admit a harmonic extension to a whole neighbourhood of the domain $P$. Note that if $P\Subset U\Subset \Omega$, where $P\Subset U$ means $\overline{P}\subset U$, for bounded $C^{1,1}$-domains $P,U,\Omega\subset \mathbb{R}^3$, then every element of $L^2\mathcal{H}(U)=\{B\in L^2\mathcal{V}(U)|\operatorname{curl}(B)=0=\operatorname{div}(B)\}$ satisfies $B\in C^\infty\mathcal{V}(\overline{P})$ by elliptic regularity and hence $B|_P$ has finite $C^k$-norm for every $k\in \mathbb{N}$. Similarly, $\operatorname{Im}(\operatorname{BS}_{\partial\Omega}^P)\subset C^\infty\mathcal{V}(\overline{P})$ and so any element in the image of the Biot-Savart operator has a finite $C^k$-norm on $P$.
\newline
With this in mind, we formulate the following result
\begin{cor}
	\label{3C10}
	Let $P\Subset U \Subset \Omega\subset \mathbb{R}^3$ be bounded $C^{1,1}$-domains with $\overline{P}\cong D^2\times S^1\cong \overline{\Omega}$ and let $S:=\partial\Omega$. If $\Omega$ admits a poloidal closed $C^1$-curve $\gamma$ contained in $S$ which bounds a $C^2$-disc $D$ in $\Omega$ such that $D\cap P$ is yet again a $C^2$-disc bounded by a poloidal closed $C^1$-curve $\widetilde{\gamma}$ contained in $\partial P$, then for every $B_T\in L^2\mathcal{H}(U)$, $k\in \mathbb{N}_0$ and every $\epsilon>0$, there exists some $j\in L^2\mathcal{V}_0(S)$, such that
	\[
	\|\operatorname{BS}_S^P(j)-B_T\|_{C^k(P)}\leq \epsilon.
	\]
\end{cor}
\begin{proof}[Proof of \Cref{3C10}]
	The main idea is to exploit elliptic estimates together with Sobolev regularity results. We can fix any $C^\infty$-smooth vector field $X$ on $\mathbb{R}^3$ which is everywhere outward pointing along $\partial P$. Multiply this vector field by a bump function so that wlog $X$ is identically zero outside a small neighbourhood of $\partial P$. We can then consider for small $t>0$, $P_t:=\psi_t(P)$, where $\psi_t$ is the flow of $X$. We observe that $\overline{P}\subset P_t$ for all $t>0$ and that if $D$ is the disc within $\Omega$ bounded by a poloidal curve $\gamma$ such that $D\cap P$ is bounded by a poloidal curve $\tilde{\gamma}$ in $\partial P$, then $\psi_t(D)$ will be a disc within $\Omega$ (since $\psi_t(\Omega)=\Omega$) which remains bounded by the same poloidal curve $\gamma$ (because $\psi_t$ is the identity outside a small neighbourhood of $\partial P$) and we have $\psi_t(D)\cap P_t=\psi_t(D\cap P)$ which will be a disc in $P_t$ which is bounded by the curve $\psi_t\circ \widetilde{\gamma}$ which remains a poloidal curve for small $t$ and is contained in $\partial P_t$.
	
We now observe that any $B\in L^2\mathcal{H}(U)$ is in fact, by elliptic estimates, analytic in $U$. Further it satisfies $\Delta B=0$ in the classical sense on $U$. Then, since $\overline{P}\subset P_t$, we can exploit elliptic estimates (keeping in mind that $\Delta B=0$), \cite[Chapter 6.3 Theorem 2]{Evans10},\cite[Theorem 9.11]{GT01} and Morrey's inequality \cite[Chapter 5.6 Theorem 6]{Evans10} to conclude
\begin{gather}
	\nonumber
	\|B\|_{C^k(P)}\leq c_1\|B\|_{W^{{k+2},2}(P)}\leq c_2\|B\|_{L^2\left(P_{t}\right)},
\end{gather}
where the constants $c_i>0$ depend on $k$,$P$ and the fixed small value $t$. We recall that $\operatorname{BS}_S^P=\operatorname{BS}_S^U|_P$ and that $\operatorname{Im}(\operatorname{BS}_S^U)\subset L^2\mathcal{H}(U)$ so that we obtain
\[
\|\operatorname{BS}^P_S(j)-B_T\|_{C^k(P)}\leq c_2\|\operatorname{BS}^{P_t}_S(j)-B_T\|_{L^2(P_t)}\text{ for all }j\in L^2\mathcal{V}_0(S).
\]
We have however already verified that $P_t$ satisfies the conditions of \Cref{3C9} so that we know that $\operatorname{Im}(\operatorname{BS}^{P_t}_S)$ is $L^{2}(P_t)$-dense in $L^{2}\mathcal{H}(P_t)$. But obviously $B_T|_{P_t}\in L^{2}\mathcal{H}(P_t)$ and so there exists some $j_{\epsilon}\in L^2\mathcal{V}_0(S)$ with $\|\operatorname{BS}_S^{P_t}(j_{\epsilon})-B_T\|_{L^{2}(P_t)}\leq \frac{\epsilon}{c_2}$ which concludes the proof.
\end{proof}
\section{Relation to optimisation problems in plasma physics}
In the plasma physics literature it is customary to use Tikhonov-type regularisation methods to make ill-posed problems accessible, c.f. \cite{L17}, \cite{PRS22}. More precisely, given some target field $B_T\in L^2\mathcal{H}(P)$, where $\overline{P}\subset \Omega\subset \mathbb{R}^3$ is our plasma domain one would like to find a current $j\in L^2\mathcal{V}_0(S)$, $S=\partial\Omega$, which minimises the quantity $\|\operatorname{BS}_S(j)-B_T\|_{L^2(P)}$. However, in order to guarantee the existence of a minimising current one introduces a regularising parameter $\lambda>0$ and can study the following regularised minimisation problem.
\begin{gather}
	\label{4E1}
	C(\lambda;B_T):=\inf_{j\in L^2\mathcal{V}_0(S)}\left(\|\operatorname{BS}_S(j)-B_T\|^2_{L^2(P)}+\lambda\|j\|^2_{L^2(S)}\right).
\end{gather}
It follows then from standard variational techniques and convexity of the functional that a unique minimiser always exists.

In this context the question of the behaviour of $C(\lambda;B_T)$ arises as $\lambda\searrow 0$. At the very least one would like to understand if this quantity approaches zero or not, which has close connections to questions regarding the approximation properties of the Biot-Savart operator.
Let us make this statement precise.
\begin{prop}
	\label{4P1}
	Let $P,\Omega\subset \mathbb{R}^3$ be bounded $C^{1,1}$-domains, $S:=\partial\Omega$, with $\overline{P}\subset \Omega$. Given some $B_T\in L^2\mathcal{H}(P)$ the following two statements are equivalent
	\begin{enumerate}
		\item $\lim_{\lambda \searrow 0}C(\lambda;B_T)=0$.
		\item There exists a sequence $(j_n)_n\subset L^2\mathcal{V}_0(S)$ with $\|\operatorname{BS}_S(j_n)-B_T\|_{L^2(P)}\rightarrow 0$ as $n\rightarrow\infty$.
	\end{enumerate}
	In particular, we have $\lim_{\lambda \searrow 0}C(\lambda;B_T)=0$ for every $B_T\in L^2\mathcal{H}(P)$ if and only if $\operatorname{Im}(\operatorname{BS}_S^P)$ is $L^2(P)$-dense in $L^2\mathcal{H}(P)$.
\end{prop}
\begin{proof}[Proof of \Cref{4P1}]
	$\quad$
	\newline
	\newline
	\underline{(i)$\Rightarrow$ (ii):} We assume for a contradiction that there does not exist a sequence $(j_n)_n$ for which $\operatorname{BS}_S^P(j_n)$ converges to $B_T$ in $L^2(P)$. Then the $L^2(P)$-orthogonal complement of $\operatorname{Im}(\operatorname{BS}_S^P)$, recall that $L^2\mathcal{H}(P)$ is a Hilbert space, is non-empty and we have the direct sum-decomposition $L^2\mathcal{H}(P)=\operatorname{clos}\left(\operatorname{Im}(\operatorname{BS}_S^P)\right)\oplus \left(\operatorname{Im}(\operatorname{BS}_S^P)\right)^\perp$. In particular, we can express $B_T=B_1+B_2$ for an element $B_1\in \operatorname{clos}\left(\operatorname{Im}(\operatorname{BS}_S^P)\right)$ and $B_2\in \left(\operatorname{Im}(\operatorname{BS}_S^P)\right)^\perp$. Since we assume that $B_T\notin \operatorname{clos}\left(\operatorname{Im}(\operatorname{BS}_S^P)\right)$ we must have $B_2\neq 0$. Consequently
	\begin{gather}
		\nonumber
		\|\operatorname{BS}_S(j)-B_T\|^2_{L^2(P)}+\lambda \|j\|^2_{L^2(S)}\geq \|B_2\|^2_{L^2(P)}\text{ for all }\lambda>0\text{ and }j\in L^2\mathcal{V}_0(S)
	\end{gather}
	by $L^2(P)$-orthogonality of the decomposition. Thus $C(\lambda;B_T)\geq \|B_2\|^2_{L^2(P)}>0$ which contradicts $\lim_{\lambda\searrow 0}C(\lambda;B_T)=0$.
	\newline
	\newline
	\underline{(ii)$\Rightarrow$(i):} Given any $\epsilon>0$, fix some $j_{\epsilon}\in L^2\mathcal{V}_0(S)$ with $\|\operatorname{BS}_S^P(j_{\epsilon})-B_T\|^2_{L^2(P)}\leq \epsilon$. Then for any fixed $\lambda>0$ we have
	\[
	C(\lambda;B_T)\leq\|\operatorname{BS}_S(j_{\epsilon})-B_T\|^2_{L^2(P)}+\lambda\|j_{\epsilon}\|^2_{L^2(S)}\leq \epsilon+\lambda \|j_{\epsilon}\|^2_{L^2(S)}.
	\]
	We note that $j_{\epsilon}$ is independent of $\lambda$ and so for the fixed $\epsilon>0$ we obtain
	\[
	0\leq \liminf_{\lambda \searrow 0}C(\lambda;B_T)\leq \limsup_{\lambda\searrow 0}C(\lambda;B_T)\leq \epsilon.
	\]
	Since $\epsilon>0$ was arbitrary, we get $\lim_{\lambda\searrow 0}C(\lambda;B_T)=0$ as desired.
\end{proof}
The following is then a direct consequence of \Cref{3C10} and \Cref{4P1}, where we recall that we write $P\Subset \Omega$ to mean $\overline{P}\subset \Omega$.
\begin{cor}
	\label{4C2}
	Let $P\Subset \Omega\subset \mathbb{R}^3$ be bounded $C^{1,1}$-domains with $\overline{P}\cong D^2\times S^1\cong \overline{\Omega}$. If $\Omega$ admits a poloidal closed $C^1$-curve $\gamma$ contained in $\partial\Omega$ which bounds a $C^2$-disc $D$ in $\Omega$ such that $D\cap P$ is yet again a $C^2$-disc bounded by a poloidal closed $C^1$-curve $\tilde{\gamma}$ contained in $\partial P$, then for every $B_T\in L^2\mathcal{H}(P)$ we have
	\[
	\lim_{\lambda\searrow 0}C(\lambda;B_T)=0.
	\]
\end{cor}
The above observation enables us to obtain a sequence of currents $(j_n)_n\subset L^2\mathcal{V}_0(S)$ whose magnetic fields $(\operatorname{BS}_S(j_n))_n$ approximate a given target field $B_T$.
\begin{cor}
	\label{4C3}
	Let $P\Subset \Omega\subset \mathbb{R}^3$ be bounded $C^{1,1}$-domains with $\overline{P}\cong D^2\times S^1\cong \overline{\Omega}$. Suppose further that $\Omega$ admits a poloidal closed $C^1$-curve $\gamma$ contained in $\partial\Omega$ which bounds a $C^2$-disc $D$ in $\Omega$ such that $D\cap P$ is yet again a $C^2$-disc bounded by a closed poloidal $C^1$-curve $\tilde{\gamma}$ contained in $\partial P$. Given any $B_T\in L^2\mathcal{H}(P)$ and $\lambda>0$ we denote by $j_{\lambda}$ the (unique) minimiser of (\ref{4E1}). Then $\lim_{\lambda \searrow 0}\|\operatorname{BS}_{\partial\Omega}(j_\lambda)-B_T\|_{L^2(P)}=0$.
\end{cor}
\begin{proof}[Proof of \Cref{4C3}]
	This is a direct consequence of the fact that $\|\operatorname{BS}_{\partial\Omega}(j_{\lambda})-B_T\|^2_{L^2(P)}\leq C(\lambda;B_T)$ which converges to zero according to \Cref{4C2}.
\end{proof}
\begin{rem}
	\label{4R4}
	\begin{enumerate}
		\item According to \Cref{3P3} we know that if $B_T\neq \operatorname{BS}_{\partial\Omega}(j)$ for every $j\in L^2\mathcal{V}_0(\partial\Omega)$ then the $L^2(\partial\Omega)$-norm of any sequence of currents $(j_n)_n\subset L^2\mathcal{V}_0(\partial\Omega)$ for which $\operatorname{BS}_{\partial\Omega}(j_n)$ converges to $B_T$ in $L^2(P)$ must diverge to infinity. For the specific sequence from \Cref{4C3} we have however the estimate
		\[
		\|j_{\lambda}\|^2_{L^2(\partial\Omega)}\leq \frac{C(\lambda;B_T)}{\lambda}
		\]
		and so in particular $\|j_{\lambda}\|_{L^2(\partial\Omega)}\in o\left(\frac{1}{\sqrt{\lambda}}\right)$ as $\lambda\searrow 0$ because $C(\lambda;B_T)\rightarrow 0$ as $\lambda\searrow0$.
		\item Given some $P\Subset U\Subset \Omega$, where $P$ and $\Omega$ satisfy the conditions of \Cref{4C3} and $U$ is some $C^{1,1}$-domain, let $B_T\in L^2\mathcal{H}(U)$. We may then just like in the proof of \Cref{3C10} enlarge the plasma domain $P$ by means of a flow out construction to a slightly larger domain $P\Subset P_t\Subset U$ and consider the minimisers $j_{\lambda}$ of (\ref{4E1}) on the underlying domain $P_t$. Then according to \Cref{4C3} $\operatorname{BS}_{\partial\Omega}(j_{\lambda})$ will converge to $B_T$ in $L^2(P_t)$ and by means of elliptic estimates in $C^1$-norm to $B_T$ on $P$. So that one may in fact obtain a $C^1$-approximating sequence. Note however that we do not obtain any specific rate of convergence as $\lambda\searrow 0$.
	\end{enumerate}
\end{rem}
\section{Kernel of $\operatorname{BS}_S$}
For a bounded $C^{1,1}$-domain $\Omega\subset \mathbb{R}^3$ we let $S:=\partial\Omega$ and consider the operator
\[
\operatorname{BS}_S:L^2\mathcal{V}_0(S)\rightarrow \mathcal{V}(\Omega)\text{, }j\mapsto \left(x\mapsto\frac{1}{4\pi}\int_Sj(y)\times \frac{x-y}{|x-y|^3}d\sigma(x)\right)
\]
where $\mathcal{V}(\Omega)$ denotes the smooth vector fields on $\Omega$ (recall that $L^2\mathcal{V}_0(S)$ is the space of $L^2(S)$-vector fields which are tangent to $S$ and divergence-free as vector fields on $S$ in the weak sense).

Before stating the rigorous version of our main result we recall that $\mathcal{H}_N(\Omega)$ is the space of $H^1$-vector fields on $\Omega$ which are div-free, curl-free and tangent to the boundary of $\Omega$ and $g(S)$ denotes the genus of a surface $S$.

The following is our main result.
\begin{thm}
	\label{3T11}
	Let $\Omega\subset \mathbb{R}^3$ be a bounded $C^{1,1}$-domain. Then
	\[
	\dim\left(\operatorname{Ker}(\operatorname{BS}_{\partial\Omega})\right)=\dim\left(\mathcal{H}_N(\Omega)\right)=g(\partial\Omega).
	\]
	In particular, $\dim\left(\operatorname{Ker}(\operatorname{BS}_{\partial\Omega})\right)$ is a homotopy invariant, i.e. for any two given bounded $C^{1,1}$-domains $\Omega_1,\Omega_2\subset\mathbb{R}^3$ which are homotopic, the kernels of the corresponding Biot-Savart operators are isomorphic. Further, $\operatorname{Ker}\left(\operatorname{BS}_{\partial\Omega}\right)\subset \bigcap_{0<\alpha<1}C^{0,\alpha}\mathcal{V}(\partial \Omega)$, i.e. all elements in the kernel are $\alpha$-H\"{o}lder continuous for all $0<\alpha<1$.
\end{thm}
\begin{rem}
	We provide an explicit isomorphism between $\mathcal{H}_N(\Omega)$ and the kernel of the Biot-Savart operator. In particular, we can construct a basis of the kernel from any given basis of $\mathcal{H}_N(\Omega)$.
\end{rem}
From the point of view of physical applications the case where a surface $S$ bounds a solid torus is of particular interest. Letting $D^2$ denote the closed unit disc of $\mathbb{R}^2$ we obtain the following corollary.
\begin{cor}
	\label{3C13}
	Let $D^2\times S^1\cong \overline{\Omega}\subset \mathbb{R}^3$, $\partial \Omega\in C^{1,1}$ be a bounded $C^{1,1}$-domain. Then
	\[
	\dim\left(\operatorname{Ker}(\operatorname{BS}_{\partial \Omega})\right)= 1.
	\]
\end{cor}
We divide the proof of \Cref{3T11} into two parts which we deal with in the following two subsections. In the first part we prove the upper bound on the dimension and in a second step we provide the lower bound.
\subsection{$\dim\left(\operatorname{Ker}(\operatorname{BS}_{\partial\Omega})\right)\leq \dim\left(\mathcal{H}_N(\Omega)\right)$}
Before we come to the proof of the upper bound we first prove two auxiliary results. The first result states that we can always find certain extensions related to square integrable currents which posses some regularity and that under certain circumstances these extensions may be taken to be curl-free. The second statement provides an equivalent expression for $\operatorname{BS}_S$ in terms of the curl of suitable extensions related to the underlying current.

Before we state the first lemma we recall that $H(\operatorname{curl},\Omega)$ is the space of square integrable vector fields which admit a square integrable curl. Further, we say that $w\in H\left(\operatorname{curl},\Omega\right)$ satisfies $\mathcal{N}\times w=j$ on $\partial\Omega$ for some $j\in L^2\mathcal{V}(\partial\Omega)$ if it satisfies
\begin{gather}
	\nonumber
	 \int_\Omega\operatorname{curl}(w)(x)\cdot \psi(x)d^3x-\int_\Omega w(x)\cdot\operatorname{curl}(\psi)(x)d^3x=\int_{\partial\Omega}j\cdot \psi d\sigma(x)\text{ for all }\psi\in W^{1,2}\mathcal{V}(\Omega).
\end{gather}
Further, we equip $H(\operatorname{curl},\Omega)$ with the norm $\|w\|_{L^2,\operatorname{curl}}:=\sqrt{\|w\|^2_{L^2(\Omega)}+\|\operatorname{curl}(w)\|^2_{L^2(\Omega)}}$ which is induced by an inner product.
\begin{lem}
	\label{4EXTRALemma}
	Let $\Omega\subset \mathbb{R}^3$ be a bounded $C^{1,1}$-domain. Then there exists a bounded linear operator
	\begin{gather}
		\nonumber
		T:L^2\mathcal{V}_0(\partial\Omega)\rightarrow H(\operatorname{curl},\Omega)
	\end{gather}
	such that $\operatorname{div}(T(j))=0$ on $\Omega$ in the weak sense and $\mathcal{N}\times T(j)=j$ on $\partial\Omega$ for all $j\in L^2\mathcal{V}_0(\partial\Omega)$. Further, if $\int_{\partial\Omega}j\cdot \Gamma d\sigma(x)=0$ for all $\Gamma\in \mathcal{H}_N(\Omega)$, then $\operatorname{curl}(T(j))=0$.
	\end{lem}
	We note that the existence of extension operators in the context of $H(\operatorname{curl},\Omega)$ spaces are well-established, see for instance \cite[Theorem 3.1]{AV96} and in fact following the construction of the extension operator in the proof of \cite[Theorem 3.1]{AV96} shows that their constructed extension operator satisfies all the required properties as stated in \Cref{4EXTRALemma}. However, to make the present manuscript more self-contained, we provide here a shorter, alternative proof exploiting the Hodge-decomposition theorem.
	\begin{proof}[Proof of \Cref{4EXTRALemma}]
		We adapt the reasoning of \cite[Corollary 2.8]{GR86} which dealt with the corresponding space $H(\operatorname{div},\Omega)$.
		
		Define the space $H^1_{T}\mathcal{V}_0(\Omega):=\{A\in H^1\mathcal{V}(\Omega)\mid \operatorname{div}(A)=0\text{, }A\cdot \mathcal{N}=0\}$ and note that it is an $H^1$-closed subspace of $W^{1,2}\mathcal{V}(\Omega)$. It then in addition follows from \cite[Lemma 2.11]{ABDG98} that the norms $\|\cdot\|_{L^2,\operatorname{curl}}$ and $\|\cdot\|_{H^1}$ are equivalent on $H^1_T\mathcal{V}_0(\Omega)$. Given some $j\in L^2\mathcal{V}_0(\partial\Omega)$ we can then define the following linear operator
		\begin{gather}
			\nonumber
			J:H^1_T\mathcal{V}_0(\Omega)\rightarrow \mathbb{R}\text{, }\psi\mapsto  \int_{\partial\Omega}j\cdot \psi d\sigma(x).
		\end{gather}
		It then follows from the trace inequality that $|J(\psi)|\leq \|j\|_{L^2(\partial\Omega)}\|\psi\|_{L^2(\partial\Omega)}\leq c_1\|j\|_{L^2(\partial\Omega)}\|\psi\|_{H^1(\Omega)}\leq c_2\|j\|_{L^2(\partial\Omega)}\|\psi\|_{L^2,\operatorname{curl}}$ where we used the equivalence of the corresponding norms and where $c_1,c_2>0$ are suitable constants which are independent of $j$ and $\psi$. It then follows from the Riesz representation theorem that there exists a unique $A_j\in H^1_T\mathcal{V}_0(\Omega)$ with
		\begin{gather}
			\label{4EEXTRA1}
			\int_{\partial\Omega}j\cdot \psi d\sigma(x)=\int_\Omega A_j(x)\cdot \psi(x)d^3x+\int_{\Omega}\operatorname{curl}(A_j)(x)\cdot \operatorname{curl}(\psi)(x)d^3x\text{ for all }\psi\in H^1_T\mathcal{V}_0(\Omega).
		\end{gather}
		It is then obvious that the assignment $j\mapsto A_j$ is linear. We then let $w_j:=-\operatorname{curl}(A_j)$ and note that the assignment $j\mapsto w_j$ is also linear. We now claim that $\operatorname{curl}(w_j)=A_j$ which will establish $w_j\in H(\operatorname{curl},\Omega)$ because $A_j\in H^1\mathcal{V}(\Omega)$. To see this let $\Psi\in \mathcal{V}_c(\Omega)$ be any smooth vector field which is compactly supported in $\Omega$. According to \Cref{BTheorem1} we can decompose $\Psi=\operatorname{curl}(B)+\Gamma+\nabla \phi$ for suitable $\phi\in W^{1,2}(\Omega)$, $B\in W^{1,2}\mathcal{V}(\Omega)$ with $B\perp \partial\Omega$ and $\Gamma\in \mathcal{H}_N(\Omega)$. We then find
		\begin{gather}
			\nonumber
			\int_{\Omega}A_j\cdot \Psi d^3x=\int_\Omega A_j\cdot (\operatorname{curl}(B)+\Gamma)d^3x
		\end{gather}
		where we used that $A_j\in H^1_T\mathcal{V}_0(\Omega)$. We further notice that $\operatorname{curl}(B)+\Gamma\in H^1_T\mathcal{V}_0(\Omega)$ (note that $\operatorname{curl}(\operatorname{curl}(B))=\operatorname{curl}(\Psi)\in L^2\mathcal{V}(\Omega)$ and $\operatorname{curl}(B)\parallel \partial\Omega$ so that $\operatorname{curl}(B)\in H^1\mathcal{V}(\Omega)$). Utilising (\ref{4EEXTRA1}) we find
		\begin{gather}
			\nonumber
			\int_{\Omega}A_j\cdot \Psi d^3x=-\int_{\Omega}\operatorname{curl}(A_j)\cdot \operatorname{curl}(\Psi)d^3x-\int_{\partial\Omega}j\cdot \nabla \phi d\sigma=-\int_{\Omega}\operatorname{curl}(A_j)\cdot \operatorname{curl}(\Psi)d^3x
		\end{gather}
		where we used that $\operatorname{curl}(\nabla \phi)=0$, that $\Psi|_{\partial\Omega}=0$ and that $j$ is divergence-free in the weak sense, i.e. $L^2(\partial\Omega)$-orthogonal to all gradient fields. This proves that $\operatorname{curl}(w_j)=A_j$. We notice that
		\begin{gather}
			\nonumber
			\|w_j\|^2_{L^2,\operatorname{curl}}=\|A_j\|^2_{L^2,\operatorname{curl}}=J(A_j)\leq c_2\|j\|_{L^2(\partial\Omega)}\|A_j\|_{L^2,\operatorname{curl}}=c_2\|j\|_{L^2(\partial\Omega)}\|w_j\|_{L^2,\operatorname{curl}}
		\end{gather}
		by means of (\ref{4EEXTRA1}) and the boundedness of the functional $J$. We are left with verifying that $\mathcal{N}\times w_j=j$ on $\partial\Omega$. But this follows in the exact same spirit as the calculation with $\Psi$ above.
		
		This proves the first part of \Cref{4EXTRALemma}. In order to establish the second part we will modify the assignment $j\mapsto w_j$ further to ensure that $\operatorname{curl}(w_j)=0$ whenever $\int_{\partial\Omega}j\cdot \Gamma d\sigma(x)=0$ for all $\Gamma\in \mathcal{H}_N(\Omega)$. The main idea of the upcoming argument, namely exploiting the Hodge-decomposition theorem, is taken from \cite[Theorem 3.1.1]{S95}. Given our $w_j$ we consider the following Hodge-decomposition of $\operatorname{curl}(w_j)$
		\begin{gather}
			\nonumber
			\operatorname{curl}(w_j)=\operatorname{grad}(\rho)+\operatorname{curl}(C)+\Gamma
		\end{gather}
		for suitable $\rho\in W^{1,2}(\Omega)$, $C\in W^{1,2}\mathcal{V}(\Omega)$, $C\perp \partial\Omega$, $\operatorname{div}(C)=0$ and $\Gamma\in \mathcal{H}_N(\Omega)$. Since this decomposition is $L^2$-orthogonal we find
		\begin{gather}
			\nonumber
			\|\operatorname{grad}(\rho)\|^2_{L^2(\Omega)}=\int_{\Omega}\operatorname{curl}(w_j)\cdot \operatorname{grad}(\rho)d^3x=0
		\end{gather}
		where by means of an approximation argument we assumed that $\rho\in W^{2,2}(\Omega)$ and thus $\nabla \rho\in W^{1,2}\mathcal{V}(\Omega)$ and where we used the fact that $\operatorname{curl}(\nabla \rho)=0$ and $\mathcal{N}\times w_j=j$ which is $L^2(\partial\Omega)$-orthogonal to all gradient fields. Similarly, using $\operatorname{curl}(\Gamma)=0$ we find
		\begin{gather}
			\nonumber
			\|\Gamma\|^2_{L^2(\Omega)}=\int_{\Omega}\operatorname{curl}(w_j)\cdot \Gamma d^3x=\int_{\partial\Omega}j\cdot \Gamma d\sigma(x)
		\end{gather}
		which equals zero whenever $\int_{\partial\Omega}j\cdot \Gamma d\sigma=0$ for all $\Gamma\in \mathcal{H}_N(\Omega)$. We conclude $\operatorname{curl}(w_j)=\Gamma+\operatorname{curl}(C)$ for some $C\in W^{1,2}\mathcal{V}(\Omega)$ with $\operatorname{div}(C)=0$ and $C\perp \partial\Omega$ and where $\Gamma\in \mathcal{H}_N(\Omega)$ is zero whenever $j$ satisfies the additional orthogonality assumption. We observe first that $\|\operatorname{curl}(C)\|_{L^2(\Omega)}\leq \|\operatorname{curl}(w_j)\|_{L^2(\Omega)}$ and that the assignment $w_j\mapsto\operatorname{curl}(C)$ is (unique) and linear. Finally we let $\mathcal{H}_D(\Omega):=\{\Theta\in H^1\mathcal{V}(\Omega)|\operatorname{curl}(\Theta)=0=\operatorname{div}(\Theta)\text{, }\Theta \perp \partial\Omega\}$ which is an $L^2$-closed subspace of $L^2\mathcal{V}(\Omega)$. It then follows, see the proof of \Cref{BTheorem1}, that there is a unique vector potential $C_j\in H^1\mathcal{V}(\Omega)$ which is $L^2$-orthogonal to $\mathcal{H}_D(\Omega)$, satisfies $\operatorname{div}(C_j)=0$ and $C_j\perp\partial\Omega$ and the a priori estimate
		\begin{gather}
			\nonumber
			\|C_j\|_{W^{1,2}(\Omega)}\leq c_3\|\operatorname{curl}(w_j)\|_{L^2(\Omega)}
		\end{gather}
		for a suitable $c_3>0$ which is independent of $j$. The assignment $j\mapsto C_j$ is therefore linear and continuous. We finally let $T(j):=w_j-C_j$. By the above arguments this yields a bounded, linear operator from $L^2\mathcal{V}_0(\partial\Omega)$ to $H(\operatorname{curl},\Omega)$ with $\operatorname{div}(T(j))=0$ and $\mathcal{N}\times T(j)=\mathcal{N}\times w_j=j$ because $C_j\perp\partial\Omega$ and $\operatorname{div}(C_j)=0$ so that $T$ satisfies the conditions of the first part of the lemma. In addition, as we have seen above, if $\int_{\partial\Omega}j\cdot \Gamma d\sigma=0$ for all $\Gamma\in \mathcal{H}_N(\Omega)$, then $\operatorname{curl}(C_j)=\operatorname{curl}(w_j)$ and consequently $\operatorname{curl}(T(j))=0$ in this case.
	\end{proof}
\begin{lem}
	\label{3L14}
	Let $\Omega\subset \mathbb{R}^3$ be a bounded $C^{1,1}$-domain. Let $j\in L^2\mathcal{V}_0(\partial\Omega)$ and let $T(j)$ be defined as in \Cref{4EXTRALemma}. Then we have the identity
	\begin{gather}
		\nonumber
	\operatorname{BS}_{\partial\Omega}(j)(x)=\frac{\int_\Omega\operatorname{curl}(T(j))(y)\times \frac{x-y}{|x-y|^3}d^3y+\nabla_x\int_{\Omega}T(j)(y)\cdot \frac{x-y}{|x-y|^3}d^3y}{4\pi}- T(j)(x)\text{ for a.e. }x\in \Omega.
	\end{gather}
	In particular, $\operatorname{BS}_{\partial\Omega}:L^2\mathcal{V}_0(\partial\Omega)\rightarrow L^2\mathcal{H}(\Omega)$ is a well-defined, linear, continuous operator.
\end{lem}
\begin{proof}[Proof of \Cref{3L14}]
	In a first step we assume that $j$ is additionally of class $H^1(\partial\Omega)$. We start by defining $v:=j\times \mathcal{N}$ and $S:=\partial\Omega$ where $\mathcal{N}\in C^{0,1}(S)$ denotes the outward pointing unit normal on $S$. Due to the Lipschitz regularity of $\mathcal{N}$ we find $v\in H^1\mathcal{V}(S)$. We observe that $v\cdot \mathcal{N}=0$ a.e. on $S$ and that
	\[
	\mathcal{N}\times v=\mathcal{N}\times (j\times \mathcal{N})=j-(\mathcal{N}\cdot j)\mathcal{N}=j,
	\]
	where we used the vector triple product rule, that $|\mathcal{N}|=1$ and that $j\cdot \mathcal{N}=0$ a.e. on $S$. In addition, if $\tilde{v}\in W^{1,2}\mathcal{V}(\Omega)$ is any $H^1$-vector field on $\Omega$ with $\tilde{v}^\parallel=v$ on $S$ (or equivalently $\mathcal{N}\times \tilde{v}=\mathcal{N}\times v=j$) then we find
	\[
	4\pi \operatorname{BS}_S(j)=\int_S\left(\mathcal{N}(y)\cdot \frac{x-y}{|x-y|^3}\right)\tilde{v}(y)-\left(\tilde{v}(y)\cdot \frac{x-y}{|x-y|^3}\right)\mathcal{N}d\sigma(y)
	\]
	once more by means of the vector triple product rule. It then follows from \cite[Theorem 3.42]{DD12} and the Hardy-Littlewood-Sobolev inequality that we may assume that $\tilde{v}\in W^{1,2}\mathcal{V}(\Omega)\cap C^{\infty}\mathcal{V}(\Omega)$.
	
	Now fix any $x\in \Omega$ and let $B_{\epsilon}(x)\Subset \Omega$ (precompact in $\Omega$). Using the integration by parts formula for Sobolev functions we find, letting $n_i:=\mathcal{N}\cdot e_i$,
	\[
	\int_S\left(\tilde{v}\cdot \frac{x-y}{|x-y|^3}\right)n_id\sigma(y)-\int_{\partial B_{\epsilon}(x)}\left(\tilde{v}\cdot \frac{x-y}{|x-y|^3}\right)\frac{y_i-x_i}{|x-y|}d\sigma(y)=\int_{\Omega\setminus B_{\epsilon}(x)}\partial_{y_i}\left(\tilde{v}\cdot \frac{x-y}{|x-y|^3}\right)d^3y,
	\]
	i.e. we have
	\[
	\int_S\left(\tilde{v}\cdot \frac{x-y}{|x-y|^3}\right)\mathcal{N}d\sigma(y)-\int_{\partial B_{\epsilon}(x)}\left(\tilde{v}\cdot \frac{x-y}{|x-y|^3}\right)\frac{y-x}{|x-y|}d\sigma(y)=\int_{\Omega\setminus B_{\epsilon}(x)}\operatorname{grad}_y\left(\tilde{v}\cdot \frac{x-y}{|x-y|^3}\right)d^3y.
	\]
	Now, $\tilde{v}$ as well as $y\mapsto \frac{x-y}{|x-y|^3}$ are smooth on $\Omega\setminus B_{\epsilon}(x)$ so we can use standard vector calculus identities to express
	\[
	\operatorname{grad}_y\left(\tilde{v}\cdot \frac{x-y}{|x-y|^3}\right)=\nabla_{\tilde{v}}\frac{x-y}{|x-y|^3}+\nabla_{\frac{x-y}{|x-y|^3}}\tilde{v}+\tilde{v}\times \operatorname{curl}\left(\frac{x-y}{|x-y|^3}\right)+\frac{x-y}{|x-y|^3}\times \operatorname{curl}(\tilde{v}),
	\]
	where the derivatives are all taken with respect to $y$. First, we observe that $\operatorname{curl}\left(\frac{x-y}{|x-y|^3}\right)=0$ because $\frac{x-y}{|x-y|^3}=\operatorname{grad}_y\left(\frac{1}{|x-y|}\right)$. In addition, as $y\mapsto\frac{x-y}{|x-y|^3}$ is smooth on $\Omega\setminus B_{\epsilon}(x)$, we have the identity
	\begin{gather}
		\nonumber
		\int_{\Omega\setminus B_{\epsilon}(x)}\nabla_{\tilde{v}}\frac{x-y}{|x-y|^3}d^3y=\int_{\Omega\setminus B_{\epsilon}(x)}\tilde{v}^i(y)\partial_{y_i}\left(\frac{x^j-y^j}{|x-y|^3}\right)d^3y\text{ }e_j
		\\
		\nonumber
		=-\int_{\Omega\setminus B_{\epsilon}(x)}\operatorname{div}(\tilde{v})\frac{x-y}{|x-y|^3}d^3y+\int_{\partial B_{\epsilon}(x)}\frac{x-y}{|x-y|^3}\left(\tilde{v}(y)\cdot \frac{x-y}{|x-y|}\right)d\sigma(y)+\int_S\frac{x-y}{|x-y|^3}\left(\tilde{v}\cdot \mathcal{N}\right)d\sigma(y),
	\end{gather}
	where we used the Einstein summation convention. Combining our considerations so far we obtain
	\begin{gather}
		\nonumber
		\int_S\left(\tilde{v}\cdot \frac{x-y}{|x-y|^3}\right)\mathcal{N}d\sigma(y)
		\\
		\label{3E6}
		=\int_{\Omega\setminus B_{\epsilon}(x)}\nabla_{\frac{x-y}{|x-y|^3}}\tilde{v}+\frac{x-y}{|x-y|^3}\times \operatorname{curl}(\tilde{v})-\operatorname{div}(\tilde{v})\frac{x-y}{|x-y|^3}d^3y+\int_S\frac{x-y}{|x-y|^3}(\tilde{v}\cdot \mathcal{N})d\sigma(y).
	\end{gather}
	On the other hand
	\begin{gather}
		\nonumber
		\int_S\tilde{v}^j(y)\left(\frac{x-y}{|x-y|^3}\cdot \mathcal{N}(y)\right)d\sigma(y)-\int_{\partial B_{\epsilon}(x)}\tilde{v}^j(y)\left(\frac{x-y}{|x-y|^3}\cdot \frac{y-x}{|x-y|}\right)d\sigma(y)
		\\
		\nonumber
		=\int_{\Omega\setminus B_{\epsilon}(x)}\operatorname{div}\left(\tilde{v}^j(y)\frac{x-y}{|x-y|^3}\right)d^3y=\int_{\Omega\setminus B_{\epsilon}(x)}\frac{x^i-y^i}{|x-y|^3}\partial_i\tilde{v}^j(y)d^3y,
	\end{gather}
	where we used that $\operatorname{div}\left(\frac{x-y}{|x-y|^3}\right)=0$ as can be verified by explicit computation. Therefore,
	\[
	\int_{S}\left(\frac{x-y}{|x-y|^3}\cdot \mathcal{N}(y)\right)\tilde{v}(y)d\sigma(y)=\int_{\Omega\setminus B_{\epsilon}(x)}\nabla_{\frac{x-y}{|x-y|^3}}\tilde{v}d^3y-\int_{\partial B_{\epsilon}(x)}\frac{\tilde{v}(y)}{|x-y|^2}d\sigma(y).
	\]
	Combining this identity with (\ref{3E6}) we arrive at
	\begin{gather}
		\nonumber
		\int_S\left(\mathcal{N}\cdot \frac{x-y}{|x-y|^3}\right)\tilde{v}(y)-\left(\tilde{v}(y)\cdot \frac{x-y}{|x-y|^3}\right)\mathcal{N}d\sigma(y)
		\\
		\nonumber
		=\int_{\Omega\setminus B_{\epsilon}(x)}\operatorname{curl}(\tilde{v})\times\frac{x-y}{|x-y|^3} +\operatorname{div}(\tilde{v})\frac{x-y}{|x-y|^3}d^3y-\int_S\frac{x-y}{|x-y|^3}(\tilde{v}\cdot \mathcal{N})d\sigma(y)-\int_{\partial B_{\epsilon}(x)}\frac{\tilde{v}(y)}{|x-y|^2}d\sigma(y).
	\end{gather}
	As for the last integral we notice that
	\[
	\int_{\partial B_{\epsilon}(x)}\frac{\tilde{v}(y)}{|x-y|^2}d\sigma(y)=\frac{1}{\epsilon^2}\int_{\partial B_{\epsilon}(x)}\tilde{v}(y)d\sigma(y)=4\pi \dashint_{\partial B_{\epsilon}(x)}\tilde{v}(y)d\sigma(y),
	\]
	where $\dashint_{\partial B_{\epsilon}(x)}\equiv \frac{1}{\operatorname{Area}({\partial B_{\epsilon}(x))}}\int_{\partial B_{\epsilon}(x)}$. Since $x\in \Omega$ is an interior point and $\tilde{v}\in C^{\infty}\mathcal{V}(\Omega)$, we see that
	\[
	\lim_{\epsilon\searrow 0}\int_{\partial B_{\epsilon}(x)}\frac{\tilde{v}(y)}{|x-y|^2}d\sigma(y)=4\pi \tilde{v}(x).
	\]
	For the remaining integral we also use the smoothness of $\tilde{v}$ as follows
	\begin{gather}
		\nonumber
		\int_{\Omega\setminus B_{\epsilon}(x)}\operatorname{curl}(\tilde{v})\times\frac{x-y}{|x-y|^3} +\operatorname{div}(\tilde{v})\frac{x-y}{|x-y|^3}d^3y
		\\
		\nonumber
		=\int_{\Omega}\operatorname{curl}(\tilde{v})\times\frac{x-y}{|x-y|^3} +\operatorname{div}(\tilde{v})\frac{x-y}{|x-y|^3}d^3y-\int_{B_{\epsilon}(x)}\operatorname{curl}(\tilde{v})\times\frac{x-y}{|x-y|^3} +\operatorname{div}(\tilde{v})\frac{x-y}{|x-y|^3}d^3y
	\end{gather}
	and since $\tilde{v}$ is smooth on $\Omega$ and hence of class $C^1(\overline{B_{\epsilon}(x)})$ we can estimate
	\[
	\left|\int_{B_{\epsilon}(x)}\operatorname{curl}(\tilde{v})\times\frac{x-y}{|x-y|^3} +\operatorname{div}(\tilde{v})\frac{x-y}{|x-y|^3}d^3y\right|\leq C(x)\int_{B_{\epsilon}(x)}\frac{1}{|x-y|^2}d^3y
	\]
	for some constant $C(x)>0$ which will depend on $x$ but can be chosen independent of $\epsilon$. Translating the integral and working in spherical coordinates we see that $\lim_{\epsilon\searrow0}\int_{B_{\epsilon}(x)}\frac{1}{|x-y|^2}d\sigma(y)=0$. Therefore, taking the limit $\epsilon\searrow 0$, we conclude that
	\begin{gather}
		\nonumber
		4\pi\operatorname{BS}_S(j)(x)=\int_{\Omega}\operatorname{curl}(\tilde{v})\times \frac{x-y}{|x-y|^3}d^3y+\int_{\Omega}\operatorname{div}(\tilde{v})\frac{x-y}{|x-y|^3}d^3y-\int_S(\tilde{v}\cdot \mathcal{N})\frac{x-y}{|x-y|^3}d\sigma(y)-4\pi \tilde{v}(x).
	\end{gather}
	We observe now that
	\begin{gather}
		\nonumber
		\int_{\Omega}\operatorname{div}(\tilde{v})\frac{x-y}{|x-y|^3}d^3y-\int_S(\tilde{v}\cdot \mathcal{N})\frac{x-y}{|x-y|^3}d\sigma(y)
		\\
		\nonumber
		=-\nabla_x\left(\int_{\Omega}\frac{\operatorname{div}(\tilde{v})(y)}{|x-y|}d^3y-\int_S\frac{(\tilde{v}(y)\cdot \mathcal{N}(y))}{|x-y|}d\sigma(y)\right)=\nabla_x\int_{\Omega}\tilde{v}(y)\cdot \frac{x-y}{|x-y|^3}d^3y
	\end{gather}
	where we used an integration by parts in the last step. Therefore we arrive at
	\begin{gather}
		\nonumber
		4\pi\operatorname{BS}_S(j)(x)=\int_{\Omega}\operatorname{curl}(\tilde{v})\times \frac{x-y}{|x-y|^3}d^3y+\nabla_x\int_{\Omega}\tilde{v}(y)\cdot \frac{x-y}{|x-y|^3}d^3y-4\pi\tilde{v}(x)\text{ for a.e. }x\in \Omega.
	\end{gather}
	We recall that $\tilde{v}\in H^1\mathcal{V}(\Omega)$ was an arbitrary element satisfying $\mathcal{N}\times \tilde{v}=j$ on $S$. According to \Cref{4EXTRALemma} we have $T(j)\in H(\operatorname{curl},\Omega)$, $\operatorname{div}(T(j))=0$ and $\mathcal{N}\times T(j)=j$. It is then a consequence of \cite[Remark 2.14 \& Corollary 2.15]{ABDG98} that in fact $T(j)\in W^{1,2}\mathcal{V}(\Omega)$ and consequently we may set $\tilde{v}=T(j)$ and deduce the desired formula for the case $j\in H^1\mathcal{V}_0(S)$. In the general case we can for given $j\in L^2\mathcal{V}_0(S)$ find an $L^2(S)$-approximating sequence $(j_n)_n\subset H^1\mathcal{V}_0(S)$. Then obviously $\operatorname{BS}_S(j)(x)=\lim_{n\rightarrow\infty}\operatorname{BS}_S(j_n)(x)$ pointwise for every fixed $x\in \Omega$. By continuity $T(j_n)$ will converge in $L^2(\Omega)$ to $T(j)$ and thus, upon passing to a subsequence, $T(j_n)$ will converge pointwise almost everywhere to $T(j)$. Similarly, the Hardy-Littlewood-Sobolev inequality implies that $\int_{\Omega}\operatorname{curl}(T(j_n))\times \frac{x-y}{|x-y|^3}d^3y$ converges in $L^6$ to $\int_{\Omega}\operatorname{curl}(T(j))\times \frac{x-y}{|x-y|^3}d^3y$ because $\operatorname{curl}(T(j_n))$ converges to $\operatorname{curl}(T(j))$ in $L^2$ by properties of the operator $T$. Lastly, if we let $N(f)(x):=\frac{1}{4\pi}\int_\Omega\frac{f(y)}{|x-y|}d^3y$ denote the Newton potential of a function $f\in L^p(\Omega)$, then we observe that
	\begin{gather}
		\nonumber
		\partial_{x^i}\int_{\Omega}T(j)\cdot \frac{x-y}{|x-y|^3}d^3y=-\partial_{x^i}\partial_{x^k}\int_{\Omega}\frac{(T(j))^k(y)}{|x-y|}d^3y=-\partial_{x^i}\partial_{x^k}N((T(j))^k)
	\end{gather}
	where $(T(j))^k$ denotes the $k$-th component of $T(j)$. It then follows from the linearity and regularity of the Newton potential \cite[Theorem 9.9]{GT01} that $\nabla_x\int_{\Omega}T(j_n)\cdot \frac{x-y}{|x-y|^3}d^3y$ converges to $\nabla_x\int_{\Omega}T(j)\cdot \frac{x-y}{|x-y|^3}d^3y$ in $L^2(\Omega)$ because $T(j_n)$ converges to $T(j)$ in $L^2(\Omega)$. Hence, upon passing to subsequences, the involved terms will converge pointwise a.e. and the claim follows.
\end{proof}
We can now combine \Cref{4EXTRALemma} and \Cref{3L14} to prove the following main ingredient needed to obtain the upper bound on the dimension of the kernel of the Biot-Savart operator.
\begin{cor}
	\label{3C16}
	Let $\Omega\subset \mathbb{R}^3$ be a bounded $C^{1,1}$-domain. If $j\in \operatorname{Ker}\left(\operatorname{BS}_{\partial\Omega}\right)$ and if
	\[
	\int_{\partial\Omega}j\cdot \Gamma d\sigma(y)=0\text{ for all }\Gamma\in\mathcal{H}_N(\Omega),
	\]
	then $j=0$.
\end{cor}
\begin{proof}[Proof of \cref{3C16}]
	Given $j\in \operatorname{Ker}(\operatorname{BS}_{\partial\Omega})$ with $\int_{\partial\Omega}j\cdot \Gamma d\sigma=0$ for all $\Gamma\in \mathcal{H}_N(\Omega)$ we denote by $T(j)\in H(\operatorname{curl},\Omega)$ the vector field from \Cref{4EXTRALemma} which satisfies $\mathcal{N}\times T(j)=j$ on $\partial\Omega$ and $\operatorname{curl}(T(j))=0$ in $\Omega$ by our orthogonality assumption on $j$. It then follows from \Cref{3L14} that
	\[
	4\pi\operatorname{BS}_{\partial\Omega}(j)=\int_{\Omega}\operatorname{curl}(T(j))\times \frac{x-y}{|x-y|^3}d^3y+\nabla_x\int_{\Omega}T(j)\cdot\frac{x-y}{|x-y|^3}d^3y-4\pi T(j)\text{ in }\Omega.
	\]
	Since $\operatorname{curl}(T(j))=0$ and $j\in \operatorname{Ker}\left(\operatorname{BS}_{\partial\Omega}\right)$ we obtain the identity
	\begin{gather}
		\label{EQUATIONextra2}
	T(j)(x)=\frac{1}{4\pi}\nabla_x\int_\Omega T(j)(y)\cdot \frac{x-y}{|x-y|^3}d^3y \text{, }x\in \Omega.
	\end{gather}
	We set $\dot{W}^{1,2}(\mathbb{R}^3):=\{f\in L^6(\mathbb{R}^3)|\operatorname{grad}(f)\in L^2\mathcal{V}(\mathbb{R}^3)\}$ and $\psi(x):=\frac{1}{4\pi}\int_{\Omega}T(j)(y)\cdot \frac{x-y}{|x-y|^3}d^3y$ and note that the Hardy-Littlewood-Sobolev inequality implies that $\psi\in L^6(\mathbb{R}^3)$ while the regularising properties of the Newton potential imply that $\nabla_x\psi(x)\in L^2\mathcal{V}(\mathbb{R}^3)$ and consequently $\psi\in \dot{W}^{1,2}(\mathbb{R}^3)$.
	
	Our goal now will be to show that we must have $T(j)\perp \partial\Omega$, i.e. $\mathcal{N}\times T(j)=0$, which will imply $j=\mathcal{N}\times T(j)=0$.
	
	Before we give a proof of the fact that $T(j)\perp \partial\Omega$ let us refer the reader to the proofs of \cite[Proposition 2 \& Theorem B]{CDG01} where a similar statement in a more regular setting was derived. The following is an alternative, shorter argument, to establish the same result which works for less regular vector fields.
	\newline
	\newline
	Let $E:=-\nabla \psi$ and $\phi\in C^1_c(\mathbb{R}^3)$ be any fixed compactly supported function. We can then approximate $j$ by vector fields $(j_n)_n\subset H^1\mathcal{V}_0(S)$ in $L^2(S)$. Letting $\psi_n(x):=\int_{\Omega}T(j_n)\cdot \frac{x-y}{|x-y|^3}d^3y$ and $E_n:=-\nabla_x\psi_n$ we know, due to the properties of the Newton potential \cite[Theorem 9.9]{GT01}, that the $E_n$ converge to $E$ in $L^2(\mathbb{R}^3)$. Then using the Fubini-Tonelli theorem and arguing in the same spirit as at the end of the proof of \Cref{3T7} we find
	\begin{gather}
		\nonumber
		\int_{\mathbb{R}^3}\nabla \phi \cdot E_nd^3x=-\int_{\Omega}\operatorname{div}(T(j_n))(y)\left(\frac{1}{4\pi}\int_{\mathbb{R}^3}\nabla \phi(x)\cdot \frac{x-y}{|x-y|^3}d^3x\right)d^3y
		\\
		\nonumber
		+\int_{\partial\Omega}\left(\mathcal{N}(y)\cdot T(j_n)(y)\right)\left(\frac{1}{4\pi}\int_{\mathbb{R}^3}\nabla \phi(x)\cdot \frac{x-y}{|x-y|^3}d^3x\right)d\sigma(y)
		\\
		\nonumber
		=\int_{\Omega}\operatorname{div}(T(j_n))(y)\phi(y)d^3y-\int_{\partial\Omega}(\mathcal{N}(y)\cdot T(j_n)(y))\phi(y)d\sigma(y)=-\int_{\Omega}T(j_n)(y)\cdot \nabla \phi(y)d^3y,
	\end{gather}
	where we first used an integration by parts, differentiated under the integral sign and made use of the fact that $T(j_n)\in H^1\mathcal{V}(\Omega)$ as explained in the proof of \Cref{3L14}. Taking the limit we conclude
	\begin{gather}
		\nonumber
		\int_{\mathbb{R}^3}\nabla \phi \cdot Ed^3x=-\int_{\Omega}T(j)(x)\cdot \nabla \phi(x)d^3x\text{ for all }\phi\in C^1_{c}(\mathbb{R}^3).
	\end{gather}
	This identity in combination with the fact that $C^1_c(\mathbb{R}^3)$ is dense in $\dot{W}^{1,2}(\mathbb{R}^3)$, in the sense that for every $f\in \dot{W}^{1,2}(\mathbb{R}^3)$ there is a sequence $(f_n)_n\subset C^1_c(\mathbb{R}^3)$ with $f_n\rightarrow f$ in $L^6(\mathbb{R}^3)$ and $\operatorname{grad}(f_n)\rightarrow \operatorname{grad}(f)$ in $L^2(\mathbb{R}^3)$, yields, by setting $\phi=\psi$,
	\[
	-\|E\|^2_{L^2(\mathbb{R}^3)}=\int_{\Omega}T(j)\cdot Ed^3x,
	\]
	where we used that $E=-\nabla \psi$. Now (\ref{EQUATIONextra2}) yields $\|E\|^2_{L^2(\mathbb{R}^3\setminus \overline{\Omega})}=0$, i.e. $-\nabla_x\psi=E=0$ in $\overline{\Omega}^c$. Hence, $\psi$ is constant on each connected component of $\overline{\Omega}^c$. Since $\psi\in \dot{W}^{1,2}(\mathbb{R}^3)$, this in particular implies that the trace of $\psi$ when viewed as a function on $\Omega$ is a locally constant function on $\partial\Omega$. In addition, $\nabla \psi=T(j)$ on $\Omega$ from which one easily infers that $j=\mathcal{N}\times T(j)=0$ on $\partial\Omega$.
\end{proof}
Let us now state the desired upper bound as a stand alone result.
\begin{cor}
	\label{3C17}
	Let $\Omega\subset \mathbb{R}^3$ be a bounded $C^{1,1}$-domain. Then $\dim\left(\operatorname{Ker}(\operatorname{BS}_{\partial\Omega})\right)\leq \dim\left(\mathcal{H}_N(\Omega)\right)$.
\end{cor}
\begin{proof}[Proof of \Cref{3C17}]
	Fix any basis $\Gamma_1,\dots,\Gamma_{\dim\left(\mathcal{H}_N(\Omega)\right)}$ of $\mathcal{H}_N(\Omega)$. We can then restrict each $\Gamma_i$ to the boundary to obtain well-defined vector fields on $\partial\Omega$. These vector fields will span a subspace of the square integrable vector fields on $\partial\Omega$. We denote this subspace by $\mathcal{H}_N(\Omega)|_{\partial\Omega}$ and we set $n:=\dim\left(\mathcal{H}_N(\Omega)|_{\partial\Omega}\right)$. We note that $n\leq \dim\left(\mathcal{H}_N(\Omega)\right)$. We now suppose that we are given $j_i\in \operatorname{Ker}(\operatorname{BS}_{\partial\Omega})$, $1\leq i\leq n+1$. Our goal is to show that any such collection of vector fields must be linearly dependent.

	If $j_1=0$ we are done. If $j_1\neq 0$ we define $V_1:=\left\{\Gamma\in \mathcal{H}_N(\Omega)|_{\partial\Omega}|\langle \Gamma,j_1\rangle_{L^2(\partial\Omega)}=0\right\}$ and we claim that $\dim\left(V_1^\perp\right)=1$, where $V_1^\perp$ denotes the $L^2(\partial\Omega)$-orthogonal complement of $V_1$ within $\mathcal{H}_N(\Omega)|_{\partial\Omega}$. To see this we first note that if $\Gamma_1,\Gamma_2\in V^\perp_1\setminus \{0\}$, then by definition of $V_1$ we must have $\int_{\partial\Omega}j_1\cdot \Gamma_kd\sigma\neq 0$ for $k=1,2$ and consequently setting $\Gamma:=\int_{\partial\Omega}j_1\cdot \Gamma_2d\sigma\Gamma_1-\int_{\partial\Omega}j_1\cdot \Gamma_1d\sigma\Gamma_2\in V_1^\perp$ we find $\int_{\partial\Omega}j_1\cdot \Gamma d\sigma=0$ so that $\Gamma\in V_1$. Therefore $\Gamma\in V_1\cap V_1^\perp=\{0\}$ and hence $\Gamma_1$ and $\Gamma_2$ are linearly dependent. On the other hand, since $j_1\neq 0$ lies in the kernel of the Biot-Savart operator, it follows from \Cref{3C16} that $\dim(V^\perp_1)\geq 1$ and so overall this space is precisely $1$-dimensional. We can now fix any $\Gamma_1\in V^\perp_1$ with $\int_{\partial\Omega}j_1\cdot \Gamma_1d\sigma=1$ spanning this space. We can then consider $j^1_m:=j_m-\int_{\partial\Omega}j_m\cdot \Gamma_1d\sigma j_1$ for $2\leq m\leq n+1$ and observe that $\langle j^1_m,\Gamma_1\rangle_{L^2(\partial\Omega)}=0$ and that $\dim\left(V_1\right)=n-1$.
	
	Now if $j^1_2=0$ this means that $j_2$ and $j_1$ are linearly dependent and we are done. So we may assume $j^1_2\neq 0$. In that case we can repeat the above procedure by setting $V_2:=\{\Gamma\in V_1|\langle \Gamma,j^1_2\rangle_{L^2(\partial\Omega)}=0\}$ and we let $V_2^\perp$ denote the $L^2(\partial\Omega)$-orthogonal complement of $V_2$ within $V_1$. It now similarly follows that $\dim\left(V^\perp_2\right)=1$ and we can fix some $\Gamma_2\in V^\perp_2\subset V_1$ with $\int_{\partial\Omega}j^1_2\cdot \Gamma_2d\sigma=1$. We then define accordingly $j^2_m:=j^1_m-\int_{\partial\Omega}j^1_m\cdot \Gamma_2d\sigma j^1_2$ for $3\leq m\leq n+1$. It follows accordingly that $\langle j^2_m,\Gamma_k\rangle_{L^2(\partial\Omega)}=0$ for all $3\leq m\leq n+1$ and $1\leq k\leq 2$.
	
	Again, if $j^2_3=0$ we see that $j_1,j_2,j_3$ must have been linearly dependent and we are done. So overall either $j_1,\dots,j_n$ are linearly dependent (in which case we are done) or otherwise, after repeating the above procedure $n$-times we obtain an $L^2(\partial\Omega)$-orthogonal basis $\Gamma_1,\dots \Gamma_n$ of $\mathcal{H}_N(\Omega)|_{\partial\Omega}$ and suitable constants $\alpha_i\in \mathbb{R}$, $1\leq i\leq n$ such that $j_{n+1}-\sum_{i=1}^n\alpha_ij_i$ is $L^2(\partial\Omega)$-orthogonal to all $\Gamma_k$, $1\leq k\leq n$. Hence, setting $j:=j_{n+1}-\sum_{i=1}^n\alpha_ij_i$, $\int_{\partial\Omega}j\cdot \Gamma d\sigma=0$ for all $\Gamma \in \mathcal{H}_N(\Omega)$ and since $j$ is in the kernel of the Biot-Savart operator (recall $\operatorname{BS}_{\partial\Omega}$ is a linear operator) it follows from \Cref{3C16} that we must have $j=0$ and hence in any case $j_1,\dots,j_{n+1}$ must be linearly dependent. We conclude $\dim\left(\operatorname{Ker}(\operatorname{BS}_{\partial\Omega})\right)\leq \dim\left(\mathcal{H}_N(\Omega)|_{\partial\Omega}\right)\leq \dim\left(\mathcal{H}_N(\Omega)\right)$.
\end{proof}
\subsection{$\dim\left(\operatorname{Ker}(\operatorname{BS}_{\partial\Omega})\geq \dim\left(\mathcal{H}_N(\Omega)\right)\right)$}
\begin{prop}
	\label{3P18}
	Let $\Omega\subset \mathbb{R}^3$ be a bounded $C^{1,1}$-domain. Then $\dim\left(\operatorname{Ker}(\operatorname{BS}_{\partial\Omega})\right)\geq \dim\left(\mathcal{H}_N(\Omega)\right)$.
\end{prop}
The proof of \Cref{3P18} consists of 4 steps. To ease the exposition we will assume at some instances in steps 1--3 that $\partial\Omega$ is connected. If $\partial\Omega$ is disconnected some additional work is required which we postpone to step 4.

An outline of the four steps is as follows. In the first step we define for every non-zero $\Gamma\in \mathcal{H}_N(\Omega)$ an explicit non-zero element $j_0\in L^2\mathcal{V}_0(\partial\Omega)$ which we show to be of class $C^{0,\alpha}\mathcal{V}(\partial\Omega)$ for all $0<\alpha<1$. In the upcoming step 2 we then verify that the defined current $j_0$ lies in the kernel of the Biot-Savart operator. In the third step we then show that the non-zero currents obtained in this way from a basis of $\mathcal{H}_N(\Omega)$ are linearly independent, which will prove the claim. In the last step we explain how to adjust the argument if the boundary is disconnected.

The regularity claim of \Cref{3T11} then comes for free from \Cref{3C17} because the upper bound then guarantees that the linearly independent vector fields which are of class $C^{0,\alpha}$ for every $0<\alpha<1$ in fact form a basis of $\operatorname{Ker}(\operatorname{BS}_{\partial\Omega})$. We state this result as a corollary.
\begin{cor}
	\label{3C19}
	Let $\Omega\subset \mathbb{R}^3$ be a bounded $C^{1,1}$-domain. Then $\operatorname{Ker}(\operatorname{BS}_{\partial\Omega})\subset \bigcap_{0<\alpha<1}C^{0,\alpha}\mathcal{V}(\partial\Omega)$.
\end{cor}
\begin{proof}[Proof of \Cref{3P18}]
	$\quad$
	\newline
	\newline
	\underline{Step 1:}
	\newline
	Before we define the current $j_0$ we fix an orientation on $\partial\Omega$ in such a way that given any vector field $X$ tangent to $\partial\Omega$ the corresponding orthogonal field $X^\perp$ obtained from $X$ by identifying $X$ with a $1$-form $\omega^1_X$ via the pulled back Euclidean metric and then identifying the $1$-form $\star \omega^1_X$ with $X^\perp$ is given by $X^\perp=\mathcal{N}\times X$ where $\mathcal{N}$ denotes the outward pointing unit normal.
	
	We now fix some $\Gamma\in \mathcal{H}_N(\Omega)$ with $\|\Gamma\|^2_{L^2(\Omega)}=1$ where we recall that $\mathcal{H}_N(\Omega)$ is the space of $H^1(\Omega)$-vector fields which are tangent to the boundary and curl- and div-free. It follows from \Cref{ALemma1} that $\Gamma\in W^{1,p}\mathcal{V}(\Omega)$ for all $1<p<\infty$. In particular, by standard Sobolev embeddings, the restriction $\Gamma|_{\partial\Omega}$ is of class $C^{0,\alpha}\mathcal{V}(\partial\Omega)$ for all $0<\alpha<1$.
	
	We can now consider the volume Biot-Savart potential $\operatorname{BS}_{\Omega}(\Gamma)(x):=\frac{1}{4\pi}\int_{\Omega}\Gamma(y)\times \frac{x-y}{|x-y|^3}d^3y$ and it follows from the proof of \Cref{3P6} that we have the identity
	\[
	\operatorname{BS}_{\Omega}(\Gamma)(x)=-\frac{1}{4\pi}\int_{\partial\Omega}\frac{\mathcal{N}(y)\times \Gamma(y)}{|x-y|}d\sigma(y)\text{ for  }x\in \Omega,
	\]
	where we used that $\Gamma$ is curl-free. Our previous arguments imply that $\mathcal{N}\times \Gamma$ is a $C^{0,\alpha}\mathcal{V}(\partial\Omega)$ vector field for every $0<\alpha<1$ and it is then standard that the right hand side of the above identity in fact defines a continuous vector field on all of $\mathbb{R}^3$. Then, in turn, \cite[Theorem 4.17]{RCM21} implies that $\nabla_x \operatorname{BS}_{\Omega}(\Gamma)|_{\Omega}$ admits a unique continuous extension up to the boundary which is of H\"{o}lder class $C^{0,\alpha}(\overline{\Omega})$ for every $0<\alpha<1$. Because $\operatorname{BS}_{\Omega}(\Gamma)$ is also continuous on all of $\mathbb{R}^3$ this then implies that $\operatorname{BS}_{\Omega}(\Gamma)|_{\partial\Omega}\in C^{1,\alpha}(\partial\Omega)$ for all $0<\alpha<1$. 
	
	Lastly, we may define
	\begin{equation}
		\label{3E9}
		F:\Omega\rightarrow\mathbb{R},x\mapsto \frac{1}{4\pi}\int_{\partial\Omega}\frac{\operatorname{BS}_{\Omega}(\Gamma)(y)\cdot \mathcal{N}(y)}{|x-y|}d\sigma(y)
	\end{equation}
	and we note that since $\operatorname{BS}_{\Omega}(\Gamma)\cdot \mathcal{N}$ is of class $C^{0,\alpha}(\partial\Omega)$ for each $0<\alpha<1$ we can argue identically as in the case of the Biot-Savart operator that $F$ is in fact continuous up to the boundary of $\Omega$ and that its restriction to the boundary is of class $C^{1,\alpha}(\partial\Omega)$ for every $0<\alpha<1$. Further, setting
	\[
	w_{\Omega}:C^{1,\alpha}(\partial\Omega)\rightarrow C^{1,\alpha}(\Omega),g\mapsto w_{\Omega}[g](x):=\frac{1}{4\pi}\int_{\partial\Omega}g(y)\frac{y-x}{|x-y|^3}\cdot \mathcal{N}(y)d\sigma(y)\text{, }x\in \Omega
	\]
	it follows from \cite[Theorem 4.31]{RCM21} that $w_{\Omega}$ is well-defined and that in particular $w_\Omega[g]$ admits a $C^{1,\alpha}(\overline{\Omega})$ extension denoted by $w^+_\Omega[g]$ (we follow here the notation in \cite{RCM21}).
	
	So far we did not need to assume the connectedness of $\partial\Omega$. But for the moment, to ease the exposition, we assume that $\partial\Omega$ is connected. In step 4 we will deal with the disconnected case separately.
	
	If $\partial\Omega$ is connected it follows from the jump relations \cite[Theorem 6.6]{RCM21} and \cite[Corollary 6.15]{RCM21} (keep in mind that then $\mathbb{R}^3\setminus \overline{\Omega}$ is connected and hence the condition $\kappa^-=0$ in the statement of \cite[Corollary 6.15]{RCM21} is satisfied) that there exists a unique $g\in C^{1,\alpha}(\partial\Omega)$ with the following property
	\begin{equation}
		\label{3E10}
		w^+_{\Omega}[g]|_{\partial\Omega}=F|_{\partial\Omega}\text{ and }g\in C^{1,\alpha}(\partial\Omega)\text{ for all }0<\alpha<1.
	\end{equation}
	We finally define
	\begin{equation}
		\label{3E11}
		j_0:=\operatorname{BS}_\Omega(\Gamma)\times \mathcal{N}-(\nabla_{\partial\Omega}g)^\perp\in C^{0,\alpha}\mathcal{V}_0(\partial\Omega)
	\end{equation}
	for all $0<\alpha<1$, where $g$ is the unique solution of (\ref{3E10}) with $F$ defined in (\ref{3E9}) and where $\mathcal{N}$ denotes the outward pointing unit normal. We notice that by our previous arguments we find $\operatorname{BS}_\Omega(\Gamma)|_{\partial\Omega}\in C^{0,1}(\partial\Omega)$ and so $\operatorname{BS}_\Omega(\Gamma)\times \mathcal{N}\in C^{0,1}\mathcal{V}(\partial\Omega)\subset C^{0,\alpha}\mathcal{V}(\partial\Omega)$ for all $0<\alpha<1$. Further, $g\in C^{1,\alpha}(\partial\Omega)$ for all $0<\alpha<1$ and so $\nabla_{\partial\Omega}g\in C^{0,\alpha}\mathcal{V}(\partial\Omega)$ for all $0<\alpha<1$ so that the regularity claim in (\ref{3E11}) follows. In addition, in the language of differential forms, we see that $(\nabla_{\partial\Omega}g)^\perp$ corresponds to $\star d g$ which is a coexact form and hence divergence-free. Further, we have for every $f\in C^1(\partial\Omega)$
	\begin{gather}
		\nonumber
		\int_{\partial\Omega}\operatorname{grad}(f)\cdot \left(\operatorname{BS}_{\Omega}(\Gamma)\times \mathcal{N}\right)d\sigma=\int_{\partial\Omega}\mathcal{N}\cdot \left(\operatorname{grad}(f)\times \operatorname{BS}_{\Omega}(\Gamma)\right)d\sigma
		\\
		\nonumber
		=\int_\Omega\operatorname{div}\left(\operatorname{grad}(f)\times \operatorname{BS}_\Omega(\Gamma)\right)d^3x=-\int_{\Omega}\operatorname{grad}(f)\cdot \Gamma d^3x=0
	\end{gather}
	where we extended $f$ in an arbitrary way to a function $f\in C^1(\overline{\Omega})$ (notice also that strictly speaking $\mathcal{N}\cdot \left(\operatorname{grad}_{\partial\Omega}(f)\times \operatorname{BS}_\Omega(\Gamma)\right)=\mathcal{N}\cdot \left(\operatorname{grad}_{\Omega}(f)\times \operatorname{BS}_\Omega(\Gamma)\right)$ because only the tangent part contributes to this expression and one should approximate $f$ by smoother functions in $C^1$-topology to make conventionally sense of the divergence term which contains second order derivatives acting upon $f$), where we used the vector calculus identity $\operatorname{div}(X\times Y)=\operatorname{curl}(Y)\cdot X-Y\cdot \operatorname{curl}(X)$, that $\operatorname{curl}\left(\operatorname{BS}_\Omega(\Gamma)\right)=\Gamma$ because $\Gamma$ is divergence-free and tangent to the boundary and the $L^2$-orthogonality of $\mathcal{H}_N(\Omega)$ to the space of gradient fields. In particular, $j_0$ is indeed divergence-free.
	
	We now argue that $j_0\neq 0$. To this end we note first that 
	\[
	\int_{\partial\Omega}(\operatorname{grad}_{\partial\Omega}g)^\perp\cdot \Gamma d\sigma=\int_{\partial\Omega}\left(\mathcal{N}\times \operatorname{grad}(g)\right)\cdot \Gamma d\sigma=0
	\]
	where in the last step one can argue identically as previously by extending $g$ in an arbitrary $C^1$ way to $\overline{\Omega}$. On the other hand, we find
	\[
	\int_{\partial\Omega}\left(\operatorname{BS}_\Omega(\Gamma)\times \mathcal{N}\right)\cdot \Gamma d\sigma=\int_\Omega\operatorname{div}\left(\Gamma\times \operatorname{BS}_\Omega(\Gamma)\right)d^3x=-\|\Gamma\|^2_{L^2(\Omega)}=-1\neq 0
	\]
	by choice of $\Gamma$ so that altogether we see that $\int_{\partial\Omega}j_0\cdot \Gamma d\sigma\neq 0$ and thus $j_0\neq 0$.
	\newline
	\newline
	\underline{Step 2:}
	\newline
	Here we prove that $\operatorname{BS}_{\partial\Omega}(j_0)(x)=0$ for all $x\in \Omega$. To this end we recall from \Cref{3L14} that letting $v:=\mathcal{N}\times j_0$ and if $\tilde{v}\in W^{1,2}\mathcal{V}(\Omega)$ is any vector field with $\tilde{v}^\parallel=v$, then
	\begin{equation}
		\label{3E12}
		\operatorname{BS}_{\partial\Omega}(j_0)=-\operatorname{BS}_\Omega(\operatorname{curl}(\tilde{v}))-\frac{1}{4\pi}\int_\Omega \operatorname{div}(\tilde{v})\frac{x-y}{|x-y|^3}d^3y+\frac{1}{4\pi}\int_{\partial\Omega}(\tilde{v}\cdot \mathcal{N})\frac{x-y}{|x-y|^3}d\sigma(y)+\tilde{v}\text{ in }\Omega.
	\end{equation}
	Letting $j_1:=\operatorname{BS}_\Omega(\Gamma)\times \mathcal{N}$ we can apply (\ref{3E12}) with $v=\operatorname{BS}_{\Omega}(\Gamma)^\parallel$ and $\tilde{v}=\operatorname{BS}_{\Omega}(\Gamma)$ to compute
	\begin{gather}
	\nonumber
	\operatorname{BS}_{\partial\Omega}(j_1)=-\operatorname{BS}_\Omega(\Gamma)+\operatorname{BS}_\Omega(\Gamma)+\frac{1}{4\pi}\int_{\partial\Omega}\left(\operatorname{BS}_\Omega(\Gamma)\cdot \mathcal{N}\right)\frac{x-y}{|x-y|^3}d\sigma(y)
	\end{gather}
	where we used that $\operatorname{BS}_\Omega(\Gamma)$ is div-free and a vector potential of $\Gamma$. Hence
	\begin{equation}
		\label{3E13}
		\operatorname{BS}_{\partial\Omega}(j_1)=\frac{1}{4\pi}\int_{\partial\Omega}\left(\operatorname{BS}_\Omega(\Gamma)\cdot \mathcal{N}\right)\frac{x-y}{|x-y|^3}d\sigma(y).
	\end{equation}
	As for $\left(\nabla_{\partial\Omega}g\right)^\perp$ we observe that $g\in C^{1,\alpha}(\partial\Omega)$ for all $0<\alpha<1$ implies that $g\in W^{2-\frac{1}{p},p}(\partial\Omega)$ for all $1<p<\infty$ so that by means of \cite[Theorem 2.4.2.5]{Gris85} we can find a function $\tilde{g}$ of class $W^{2,p}(\Omega)$ for all $1<p<\infty$ such that $\Delta \tilde{g}=0$ in $\Omega$ and $\tilde{g}|_{\partial\Omega}=g$. It then follows that $\left(\nabla_{\partial\Omega}g\right)^\perp=\mathcal{N}\times \nabla_{\partial\Omega}g=\mathcal{N}\times \nabla_{\Omega}\tilde{g}=-\nabla \tilde{g}\times \mathcal{N}$ where we again used the fact that the normal part does not contribute when taking the cross product with $\mathcal{N}$. Hence, we may take $\tilde{v}=-\nabla \tilde{g}$ to compute the Biot-Savart operator of $\left(\nabla_{\partial\Omega}g\right)^\perp$.
	\begin{gather}
		\nonumber
	\operatorname{BS}_{\partial\Omega}(\left(\nabla_{\partial\Omega}g\right)^\perp)=-\frac{1}{4\pi}\int_{\partial\Omega}\left(\nabla \tilde{g}\cdot \mathcal{N}\right)\frac{x-y}{|x-y|^3}d\sigma(y)-\nabla \tilde{g}.
	\end{gather}
	In order to utilise (\ref{3E10}) we rewrite
	\begin{gather}
		\nonumber
	-\frac{1}{4\pi}\int_{\partial\Omega}\left(\nabla \tilde{g}\cdot \mathcal{N}\right)\frac{x-y}{|x-y|^3}d\sigma(y)=\frac{\nabla_x}{4\pi}\int_{\partial\Omega}\frac{\nabla \tilde{g}\cdot \mathcal{N}}{|x-y|}d\sigma(y)
	\end{gather}
	and upon integrating by parts twice we find
	\[
	\frac{1}{4\pi}\int_{\partial\Omega}\frac{\nabla \tilde{g}\cdot \mathcal{N}}{|x-y|}d\sigma(y)=\tilde{g}-\frac{1}{4\pi}\int_{\partial\Omega}g(y)\frac{y-x}{|y-x|^3}\cdot \mathcal{N}(y)d\sigma(y)=\tilde{g}-w_\Omega[g]\text{ in }\Omega
	\]
	where we simply plugged in the definition of $w_\Omega$ from the previous step in the last equality. We therefore arrive at
	\[
	\operatorname{BS}_{\partial\Omega}(\left(\nabla_{\partial\Omega}g\right)^\perp)=-\nabla_xw_\Omega[g]\text{ in }\Omega.
	\]
	To exploit the defining properties of $g$ we recall that $w_\Omega[g]$ admits a continuous extension $w^+_\Omega[g]$ onto $\overline{\Omega}$ and observe that $w^+_\Omega[g]\in C^0\left(\overline{\Omega}\right)\cap C^2(\Omega)$ is harmonic in $\Omega$, \cite[Proposition 4.28]{RCM21}. In addition, the function $F$ as defined in (\ref{3E9}) is in fact of class $C^0(\overline{\Omega})\cap C^2(\Omega)$ and also harmonic in $\Omega$ which follows easily by direct calculation. Therefore, by means of the maximum principle, $F$ and $w_\Omega[g]$ coincide within $\Omega$ if and only if $F$ and $w^+_\Omega[g]$ coincide on $\partial\Omega$. But this is precisely the defining property of $g$, recall (\ref{3E10}), i.e. we find
	\[
	\operatorname{BS}_{\partial\Omega}(\left(\nabla_{\partial\Omega}g\right)^\perp)=-\nabla_xF=\frac{1}{4\pi}\int_{\partial\Omega}\left(\operatorname{BS}_\Omega(\Gamma)\cdot \mathcal{N}\right)\frac{x-y}{|x-y|^3}d\sigma(y)\text{ in }\Omega.
	\]
	Combining this with (\ref{3E13}) and exploiting the linearity of the Biot-Savart operator we find $\operatorname{BS}_{\partial\Omega}(j_0)(x)=0$ for all $x\in \Omega$ as desired.
	\newline
	\newline
	\underline{Step 3:} Here we prove that if $\Gamma_1,\dots,\Gamma_{\dim\left(\mathcal{H}_N(\Omega)\right)}$ forms a basis of $\mathcal{H}_N(\Omega)$, then the currents $j_k$ obtained from the $\Gamma_k$ by means of the procedure in the previous steps are linearly independent. To this end, suppose that we are given $\alpha_i\in \mathbb{R}$ with $j:=\sum_{i=1}^{\dim\left(\mathcal{H}_N(\Omega)\right)}\alpha_{i}j_i=0$. Setting $\Gamma:=\sum_{i=1}^{\dim\left(\mathcal{H}_N(\Omega)\right)} \alpha_i\Gamma_i$ and $g:=\sum_i\alpha_ig_i$ we note that we then have by linearity
	\[
	0=j=\operatorname{BS}_{\Omega}(\Gamma)\times \mathcal{N}-\left(\nabla_{\partial\Omega}g\right)^\perp.
	\]
	It then follows identically as in the last part of step 1 that we have
	\[
	0=\langle j,\Gamma\rangle_{L^2(\partial\Omega)}=-\|\Gamma\|^2_{L^2(\Omega)}
	\]
	and consequently that $\Gamma=0$ which in turn, since the $\Gamma_i$ form a basis, implies that $\alpha_i=0$ for all $i$. Consequently the constructed currents are all linearly independent.
	\newline
	\newline
	\underline{Step 4:}
	\newline
	Here we deal with the situation where $\partial\Omega$ is disconnected. To this end we observe that $\mathbb{R}^3\setminus \Omega$ has exactly one connected component which is unbounded. The boundary of this connected component will coincide with some boundary component of $\partial\Omega$ which we denote by $\partial\Omega^0$. We set $m:=\#\partial\Omega-1$ and label the remaining boundary components of $\partial\Omega$ as $\partial\Omega^1,\dots,\partial\Omega^m$. We observe that all the arguments of step 1-3 apply verbatim as long as we can guarantee that there exists a function $g\in C^{1,\alpha}(\partial\Omega)$ satisfying (\ref{3E10}). It however follows from \cite[Theorem 5.8, Theorem 6.5, Proposition 6.13 \& Theorem 6.14]{RCM21} that for $F|_{\partial\Omega}$ where $F$ is defined in (\ref{3E9}) (recall it is continuous up to the boundary and its boundary restriction is of class $C^{1,\alpha}$), there exists a solution $g$ of (\ref{3E10}) as long as
	\[
	\int_{\partial\Omega}F(x)\cdot \mu_i(x)d\sigma(x)=0\text{, }1\leq i\leq m
	\]
	where the $\mu_i\in C^{0,\alpha}(\partial\Omega)$ satisfy $\nu[\mu_i](x):=-\frac{1}{4\pi}\int_{\partial\Omega}\frac{\mu_i(y)}{|x-y|}d\sigma(y)|_{\partial\Omega^j}=\delta_{ij}$. But this is easy to verify because $\operatorname{BS}_\Omega(\Gamma)\in H^1_{\operatorname{loc}}\mathcal{V}(\mathbb{R}^3)$ is divergence-free. So letting $\Omega^i$ denote the bounded connected component of $\mathbb{R}^3\setminus \partial\Omega^i$ we find, using Fubini's theorem and the defining properties of the $\mu_i$,
	\begin{gather}
		\nonumber
		\int_{\partial\Omega}\mu_i(x)\cdot F(x)d\sigma(x)=-\int_{\partial\Omega}\operatorname{BS}_\Omega(\Gamma)(y)\cdot \mathcal{N}(y)\nu[\mu_i](y)d\sigma(y)
		\\
		\nonumber
		=-\int_{\partial\Omega^i}\operatorname{BS}_\Omega(\Gamma)(y)\cdot \mathcal{N}(y)d\sigma(y)=\int_{\Omega^i}\operatorname{div}\left(\operatorname{BS}_{\Omega}(\Gamma)\right)(y)d^3y=0.
	\end{gather}
	Consequently, there exists a solution $g$ to (\ref{3E10}) and the remaining part of the proof applies verbatim with the only difference that the solution $g$ is no longer unique, however, the regularity assertion remains valid because any two solutions differ by an element in the kernel of an appropriate operator which consists of smooth elements \cite[Theorem 6.14]{RCM21}.
\end{proof}
\section{A recursive approximation}
Throughout this section we let $\Omega\subset \mathbb{R}^3$ be a bounded $C^{1,1}$-domain. We recall from the proof of \Cref{3P18} that if we are given some $\Gamma\in \mathcal{H}_N(\Omega)$ then we can obtain an element $j\in \operatorname{Ker}(\operatorname{BS}_{\partial\Omega})$ by setting $j:=\mathcal{N}\times \operatorname{BS}_{\Omega}(\Gamma)+\mathcal{N}\times \operatorname{grad}_{\partial\Omega}g$ where $g\in \bigcap_{0<\alpha<1}C^{1,\alpha}(\partial\Omega)$ satisfies
\begin{gather}
	\label{E61}
	\int_{\partial\Omega}g(y)\frac{y-x}{|x-y|^3}\cdot \mathcal{N}(y)d\sigma(y)=\int_{\partial\Omega}\frac{\operatorname{BS}_{\Omega}(\Gamma)\cdot \mathcal{N}}{|x-y|}d\sigma(y).
\end{gather}
We recall here that if $\partial\Omega$ is disconnected, then the solution of (\ref{E61}) is not unique, however, the induced gradient field $\operatorname{grad}_{\partial\Omega}g$ is independent of the choice of solution.

Therefore, if we are given some $\Gamma\in \mathcal{H}_N(\Omega)$, the main obstruction in finding an explicit formula for $j$ is the fact that the function $g$ is implicitly defined. The goal of this section it to reformulate the defining equation (\ref{E61}) for $g$ into a suitable fix point problem and then to provide an iterative scheme which allows us to approximate $\mathcal{N}\times \operatorname{grad}_{\partial\Omega}g$ by means of this iterative procedure in some appropriate topology.

We start by defining $\mathcal{H}_{\operatorname{ex}}(\Omega):=\{\nabla f|f\in H^1(\Omega)\text{, }\Delta f=0\}$ which is an $L^2$-closed subspace of $L^2\mathcal{H}(\Omega)$. We also recall the definition of a firmly non-expansive operator, c.f. \cite[Definition 4.1 \& Proposition 4.4]{BC19}.
\begin{defn}
	\label{6D1}
	Let $\mathcal{H}$ be a Hilbert space and $F:\mathcal{H}\rightarrow\mathcal{H}$ a map. Then we call $F$ firmly non-expansive if
	\begin{gather}
		\nonumber
		\|F(x)-F(y)\|^2\leq \langle x-y,F(x)-F(y)\rangle.
	\end{gather}
\end{defn}
We note that any firmly non-expansive map $F$ is necessarily Lipschitz-continuous with $\operatorname{Lip}(F)\leq 1$ and so in particular non-expansive.
\begin{lem}
	\label{6L2}
	Let $\Omega\subset\mathbb{R}^3$ be a bounded $C^{1,1}$-domain and $\Gamma\in \mathcal{H}_N(\Omega)\setminus \{0\}$ be any fixed element. Then the following is a well-defined, firmly non-expansive operator
	\begin{gather}
		\nonumber
		S:\mathcal{H}_{\operatorname{ex}}(\Omega)\rightarrow \mathcal{H}_{\operatorname{ex}}(\Omega)\text{, }\nabla f\mapsto \frac{\nabla}{4\pi}\int_{\partial\Omega}\frac{\operatorname{BS}_{\Omega}(\Gamma)\cdot \mathcal{N}}{|x-y|}d\sigma(y)+\frac{\nabla}{4\pi}\int_{\Omega}\nabla f(y)\cdot \frac{x-y}{|x-y|^3}d^3y.
	\end{gather}
\end{lem}
\begin{proof}[Proof of \Cref{6L2}]
	First we recall from the proof of \Cref{3P18} that $x\mapsto \int_{\partial\Omega}\frac{\operatorname{BS}_{\Omega}(\Gamma)\cdot \mathcal{N}}{|x-y|}d\sigma(y)\in \bigcap_{0<\alpha<1}C^{1,\alpha}(\overline{\Omega})$. One also verifies easily by direct calculations that this function is harmonic in $\Omega$ so that we conclude $\nabla_x \int_{\partial\Omega}\frac{\operatorname{BS}_{\Omega}(\Gamma)\cdot \mathcal{N}}{|x-y|}d\sigma(y)\in \mathcal{H}_{\operatorname{ex}}(\Omega)$. On the other hand, it follows from the regularity of the Newton potential \cite[Theorem 9.9]{GT01} that $x\mapsto \int_{\Omega}\nabla f(y)\cdot \frac{x-y}{|x-y|^3}d^3y\in H^1(\Omega)$. In addition, it follows from \Cref{3L8} that we can write
	\begin{gather}
		\nonumber
		\int_{\Omega}\nabla f(y)\cdot \frac{x-y}{|x-y|^3}d^3y=\int_{\partial\Omega}f(y)\frac{x-y}{|x-y|^3}\cdot\mathcal{N}(y)d\sigma(y)+4\pi f(x).
	\end{gather}
	One can verify by direct calculation that the first term in the above expression is harmonic in $\Omega$, while $f$ is harmonic by assumption. This proves $\nabla\int_{\Omega}\nabla f(y)\cdot \frac{x-y}{|x-y|^3}d^3y\in \mathcal{H}_{\operatorname{ex}}(\Omega)$ and establishes the well-definedness of the operator $S$.
	
	We will now prove the following identity, where we let $H:L^2\mathcal{V}(\Omega)\rightarrow H^1(\mathbb{R}^3)$, $X\mapsto \frac{1}{4\pi}\int_{\Omega}X(y)\cdot \frac{x-y}{|x-y|^3}d^3y$ and set $\bar{f}:=f_1-f_2$ for any two given functions $f_1,f_2\in H^1(\Omega)$ with $\nabla f_1,\nabla f_2\in \mathcal{H}_{\operatorname{ex}}(\Omega)$,
	\begin{gather}
		\nonumber
		\|S(\nabla f_1)-S(\nabla f_2)\|^2_{L^2(\Omega)}=\int_{\Omega}\nabla \bar{f}\cdot (S(\nabla f_1)-S(\nabla f_2))d^3x-\int_{\mathbb{R}^3\setminus \Omega}|\nabla H(\nabla \bar{f})|^2d^3x
		\\
		\label{6E2}
		\Leftrightarrow \|S(\nabla f_1)-S(\nabla f_2)\|^2_{L^2(\Omega)}=\|\nabla \bar{f}\|^2_{L^2(\Omega)}-\|\nabla \bar{f}-\nabla H(\nabla \bar{f})\|^2_{L^2(\Omega)}-2\|\nabla H(\nabla \bar{f})\|^2_{L^2(\mathbb{R}^3\setminus \Omega)}.
	\end{gather}
	In order to prove this we first observe that $S(\nabla f_1)-S(\nabla f_2)=\nabla H(\nabla \bar{f})$. We first assume that $X\in C^\infty\mathcal{V}_c(\Omega)$ is a smooth compactly supported vector field on $\Omega$. Then
	\begin{gather}
		\nonumber
		\|\nabla H(X)\|^2_{L^2(\mathbb{R}^3)}=\frac{1}{4\pi}\sum_{j,i=1}^3\int_{\mathbb{R}^3}\partial_jH(X)(x)\partial_j\int_{\Omega}X^i(y)\frac{x^i-y^i}{|x-y|^3}d^3yd^3x
		\\
		\nonumber
		=\lim_{R\rightarrow\infty}\frac{1}{4\pi}\sum_{j,i=1}^3\int_{B_R(0)}\partial_jH(X)(x)\partial_j\int_{\Omega}X^i(y)\frac{x^i-y^i}{|x-y|^3}d^3yd^3x
		\\
		\nonumber
		=\lim_{R\rightarrow\infty}\left[\frac{1}{4\pi}\sum_{i=1}^3\int_{B_R(0)}H(X)(x)\left(-\Delta\int_{\Omega}X^i(y)\frac{x^i-y^i}{|x-y|^3}d^3y\right)d^3x+\int_{\partial B_R}H(X)(x)\mathcal{N}\cdot \nabla H(X)d\sigma(x)\right].
	\end{gather}
	The boundary term vanishes in the limit due to the behaviour $|H(X)|\in \mathcal{O}\left(\frac{1}{|x|^2}\right)$ and $|\nabla H(X)|\in \mathcal{O}\left(\frac{1}{|x|^3}\right)$ as $|x|\rightarrow\infty$. In addition we have
	\begin{gather}
		\nonumber
		-\Delta \sum_{i=1}^3\int_{\Omega} X^i(y)\frac{x^i-y^i}{|x-y|^3}d^3y=-\Delta \sum_{i=1}^3\int_{\Omega}X^i(y)\partial_{y_i}\frac{1}{|x-y|}d^3y=\Delta\int_{\Omega}\frac{\operatorname{div}(X)(y)}{|x-y|}d^3y=-4\pi \operatorname{div}(X)(x)
	\end{gather}
	where we used that $X$ is compactly supported and that $-\frac{1}{4\pi|x|}$ is the fundamental solution of the Laplace operator. Using once more that $X$ is compactly supported in $\Omega$ we find
	\begin{gather}
		\nonumber
		\|\nabla H(X)\|^2_{L^2(\mathbb{R}^3)}=-\int_{\Omega}\operatorname{div}(X)(x)H(X)(x)d^3x=\int_{\Omega}X(x)\cdot \nabla H(X)(x)d^3x.
	\end{gather}
	By means of an approximation argument this identity remains valid for all $X\in L^2\mathcal{V}(\Omega)$. Letting $X=\nabla \bar{f}$ and recalling that $S(\nabla f_1)-S(\nabla f_2)=\nabla H(\nabla \bar{f})$ we immediately obtain the first line in (\ref{6E2}). The equivalent reformulation follows immediately by expanding the terms. The first identity in (\ref{6E2}) implies the firm non-expansiveness of $S$.
	\end{proof}
Before we proceed let us introduce the following notation: Given a set $M$ and a function $f:M\rightarrow M$ we define its fix point set as $\operatorname{Fix}(f):=\{x\in M|f(x)=x\}$.
\begin{lem}
	\label{6L3}
	Let $\Omega\subset \mathbb{R}^3$ be a bounded $C^{1,1}$-domain and $\Gamma\in \mathcal{H}_N(\Omega)\setminus \{0\}$ be a fixed element. Let $S$ be the corresponding operator as defined in \Cref{6L2}. Then the following holds
	\begin{enumerate}
		\item $\operatorname{Fix}(S)\neq \emptyset$.
		\item $\operatorname{Fix}(S)\subset \bigcap_{1<p<\infty}W^{1,p}\mathcal{V}(\Omega)$.
		\item If $\nabla f_1$, $\nabla f_2\in \operatorname{Fix}(S)$, then $\mathcal{N}\times \nabla f_1=\mathcal{N}\times \nabla f_2$.
		\item If $\nabla f\in \operatorname{Fix}(S)$ then $\mathcal{N}\times \nabla f=\mathcal{N}\times \nabla_{\partial\Omega}g$ where $g$ is any solution of (\ref{E61}).
	\end{enumerate} 
\end{lem}
\begin{proof}[Proof of \Cref{6L3}]
\underline{(i)} We know that (\ref{E61}) admits a solution $g\in \bigcap_{0<\alpha<1}C^{1,\alpha}(\partial\Omega)$ for any given $\Gamma\in \mathcal{H}_N(\Omega)$. We observe that then in particular $g\in W^{2-\frac{1}{p},p}(\partial\Omega)$ and that we can therefore find some $f\in \bigcap_{1<p<\infty}W^{2,p}(\Omega)$ with $\Delta f=0$ in $\Omega$ and $f|_{\partial\Omega}=g$, c.f. \cite[Theorem 2.4.2.5]{Gris85}. We observe that we can rewrite (\ref{E61}) equivalently as
\begin{gather}
	\nonumber
	\frac{1}{4\pi}\int_{\partial\Omega}\frac{\operatorname{BS}_{\Omega}(\Gamma)\cdot \mathcal{N}}{|x-y|}d\sigma(y)+\frac{1}{4\pi}\int_{\Omega}\nabla f(y)\cdot \frac{x-y}{|x-y|^3}d^3y=f(x).
\end{gather}
Taking the gradient on both sides yields $S(\nabla f)=\nabla f$ with $\nabla f\in \mathcal{H}_{\operatorname{ex}}(\Omega)\cap \bigcap_{1<p<\infty}W^{1,p}\mathcal{V}(\Omega)$. This yields $\operatorname{Fix}(S)\neq \emptyset$.
\newline
\newline
\underline{(ii):} We show now that if $\nabla f_1,\nabla f_2\in \mathcal{H}_{\operatorname{ex}}(\Omega)$ are any two fixed points of $S$, then $\nabla \bar{f}\in \mathcal{H}_D(\Omega)=\{\Theta\in L^2\mathcal{V}(\Omega)\mid\operatorname{curl}(\Theta)=0=\operatorname{div}(\Theta)\text{, }\Theta\perp\partial\Omega\}$ where $\bar{f}:=f_1-f_2$. Once this is shown, it will follow from the regularity result \Cref{ALemma2} and the fact that there exists some fixed point of class $\bigcap_{1<p<\infty}W^{1,p}\mathcal{V}(\Omega)$ as shown in the proof of (i) that all fixed points are of the desired regularity. To see that we note that if $\nabla f_1,\nabla f_2$ are any two fixed points of $S$ then $\|S(\nabla f_1)-S(\nabla f_2)\|^2_{L^2(\Omega)}=\|\nabla f_1-\nabla f_2\|^2_{L^2(\Omega)}$ and it then follows from (\ref{6E2}) that $\nabla H(\nabla \bar{f})=0$ on $\mathbb{R}^3\setminus \Omega$ and $\nabla \bar{f}=\nabla H(\nabla \bar{f})$ in $\Omega$. $\nabla H(\nabla \bar{f})=0$ on $\mathbb{R}^3\setminus \Omega$ implies together with the fact that $H(\nabla \bar{f})\in H^1(\mathbb{R}^3)$ that the trace of $H(\nabla \bar{f})$ when viewed as a function on $\Omega$ is a locally constant function on $\partial\Omega$. Since $\nabla \bar{f}=\nabla H(\nabla \bar{f})$ in $\Omega$ the same must be true for $\bar{f}$ from which one easily concludes that $\mathcal{N}\times \nabla \bar{f}=0$ and hence $\nabla \bar{f}\in \mathcal{H}_D(\Omega)$. As mentioned before this concludes the proof of statement (ii).
\newline
\newline
\underline{(iii):} This is an immediate consequence of the proof of (ii) since $\nabla f_1-\nabla f_2\in \mathcal{H}_D(\Omega)$ and consequently $\mathcal{N}\times (\nabla f_1-\nabla f_2)=0$.
\newline
\newline
\underline{(iv):} By (iii) it is enough to find one $\nabla f\in \operatorname{Fix}(S)$ with this property. We have however already seen in the proof of (i) that if $g$ is any given solution of (\ref{E61}) then the unique solution of the boundary value problem $\Delta f=0$ in $\Omega$, $f|_{\partial\Omega}=g$ gives rise to a fixed point $\nabla f\in \operatorname{Fix}(S)$. But from here it follows immediately that $\mathcal{N}\times \nabla f=\mathcal{N}\times \nabla_{\partial\Omega}f=\mathcal{N}\times \nabla_{\partial\Omega}g$.
\end{proof}
\begin{cor}
		\label{6C4}
	Let $\Omega\subset \mathbb{R}^3$ be a bounded $C^{1,1}$-domain and let $S$ be defined as in \Cref{6L2} for some fixed $\Gamma\in \mathcal{H}_N(\Omega)\setminus\{0\}$. Let further $(\lambda_n)_n\subset [0,2]$ be a sequence such that $\sum_{n=1}^\infty\lambda_n(2-\lambda_n)=+\infty$ and let $X_0\in \mathcal{H}_{\operatorname{ex}}(\Omega)$ be any fixed element. Then the sequence $(X_n)_n\subset \mathcal{H}_{\operatorname{ex}}(\Omega)$ defined by $X_{n+1}:=X_n+\lambda_n(S(X_n)-X_n)$ satisfies the following properties
	\begin{enumerate}
		\item $X_n$ converges weakly in $L^2(\Omega)$ to some $\nabla f\in \operatorname{Fix}(S)$.
		\item $\|S(X_n)-X_n\|_{L^2(\Omega)}\rightarrow 0$, $n\rightarrow\infty$.
	\end{enumerate}
\end{cor}
\begin{proof}[Proof of \Cref{6C4}]
	This is a direct consequence of \Cref{6L2}, \Cref{6L3} in combination with \cite[Corollary 5.17]{BC19}.
\end{proof}
We are now ready to state the main result of this section. To this end we define the space $W^{-\frac{1}{2},2}\mathcal{V}_0(\partial\Omega)$ as the completion of $L^2\mathcal{V}_0(\partial\Omega)$ equipped with the norm $\|j\|:=\sup_{\phi\in W^{\frac{1}{2},2}\mathcal{V}(\partial\Omega)}\frac{\left|\int_{\partial\Omega}j\cdot \phi d\sigma\right|}{\|\phi\|_{W^{\frac{1}{2},2}(\partial\Omega)}}$ and note that the Biot-Savart operator extends to a linear, continuous operator $\operatorname{BS}_{\partial\Omega}:W^{-\frac{1}{2},2}\mathcal{V}_0(\partial\Omega)\rightarrow L^2\mathcal{H}(\Omega)$, see \Cref{AppC} for more details.
\begin{thm}
	\label{6T5}
	Let $\Omega\subset \mathbb{R}^3$ be a bounded $C^{1,1}$-domain and $\Gamma \in \mathcal{H}_N(\Omega)\setminus \{0\}$ be any fixed element. Let $S:\mathcal{H}_{\operatorname{ex}}(\Omega)\rightarrow \mathcal{H}_{\operatorname{ex}}(\Omega)$ be defined as in \Cref{6L2}. Further we define the sequence $(X_n)_n\subset \mathcal{H}_{\operatorname{ex}}(\Omega)$ recursively by $X_0:=0$ and $X_{n+1}:=S(X_n)$. Then the following holds
	\begin{enumerate}
		\item $j_n:=\mathcal{N}\times \operatorname{BS}_{\Omega}(\Gamma)+\mathcal{N}\times X_n\in W^{-\frac{1}{2},2}\mathcal{V}_0(\partial\Omega)\cap \bigcap_{0<\alpha<1}C^{0,\alpha}\mathcal{V}(\partial\Omega)$ converges weakly in $W^{-\frac{1}{2},2}(\partial\Omega)$ to a non-trivial element $j\in \operatorname{Ker}(\operatorname{BS}_{\partial\Omega})$ as constructed in the proof of \Cref{3P18}.
		\item $\operatorname{BS}_{\partial\Omega}(j_n)$ converges weakly to zero in $L^2(\Omega)$.
		\item If $P\subset \Omega$ is any other $C^{1,1}$-domain with $\operatorname{dist}(P,\partial\Omega)>0$, then $\operatorname{BS}_{\partial\Omega}(j_n)$ converges strongly in $C^k(P)$ to zero as $n\rightarrow\infty$ for any fixed $k\in \mathbb{N}_0$.
	\end{enumerate}
\end{thm}
\begin{proof}[Proof of \Cref{6T5}]
	\underline{(i):}
	We start by observing that we had already shown in the proof of \Cref{3P18} that $\mathcal{N}\times \operatorname{BS}_{\Omega}(\Gamma)$ is divergence-free and of class $C^{0,1}\mathcal{V}(\partial\Omega)$. Further we claim that $S(X)\in \bigcap_{1<p<\infty}W^{1,p}\mathcal{V}(\Omega)$ whenever $X\in \bigcap W^{1,p}\mathcal{V}(\partial\Omega)\cap \mathcal{H}_{\operatorname{ex}}(\Omega)$. To this end it is enough to prove that the function $F(x):=\int_{\Omega}X(y)\cdot \frac{x-y}{|x-y|^3}d^3y$ is of class $W^{2,p}(\Omega)$ for all $1<p<\infty$. We recall that the regularity of the Newton potential implies $F\in H^1(\Omega)$ and that $\Delta F=0$ in $\Omega$. Further, since $X$ is assumed divergence-free we find $F(x)=\int_{\partial\Omega}\frac{X(y)\cdot \mathcal{N}(y)}{|x-y|}d\sigma(y)$. By Sobolev embeddings we know that $X\in \bigcap_{0<\alpha<1}C^{0,\alpha}\mathcal{V}(\overline{\Omega})$. It then follows identically as in the proof of \Cref{3P18} that $F|_{\partial\Omega}\in \bigcap_{0<\alpha<1}C^{1,\alpha}(\partial\Omega)\subseteq \bigcap_{1<p<\infty}W^{2-\frac{1}{p},p}(\partial\Omega)$. Hence it follows from standard elliptic regularity results that $F\in \bigcap_{1<p<\infty}W^{2,p}(\Omega)$ and the regularity assertion follows. The fact that all of the $(j_n)_n$ are divergence-free in the weak sense follows also from the same arguments as in the proof of \Cref{3P18} and so in particular $j_n\in W^{-\frac{1}{2},2}\mathcal{V}_0(\partial\Omega)$.
	
	Letting $\lambda_n:=1$ for all $n\in \mathbb{N}$ and setting $X_0:=0$ it follows from \Cref{6C4} that the sequence $(X_n)_n\subset \mathcal{H}_{\operatorname{ex}}(\Omega)$ converges to some $\nabla f\in \operatorname{Fix}(S)$ weakly in $L^2(\Omega)$. We observe further that $\operatorname{curl}(Y)=0$ for every $Y\in \mathcal{H}_{\operatorname{ex}}(\Omega)$ and therefore $(X_n)_n$ converges in fact weakly in $H(\operatorname{curl},\Omega)$ to $\nabla f\in \operatorname{Fix}(S)$. It then follows from the continuity of the tangential trace, see \Cref{CL1}, that $\mathcal{N}\times X_n$ converges weakly in $W^{-\frac{1}{2},2}(\partial\Omega)$ to $\mathcal{N}\times \nabla f$. According to \Cref{6L3} we have $\mathcal{N}\times \nabla f=\mathcal{N}\times \nabla_{\partial\Omega}g$ where $g$ is a solution to (\ref{E61}). We conclude that $\mathcal{N}\times \operatorname{BS}_{\Omega}(\Gamma)+\mathcal{N}\times X_n$ converges weakly to $\mathcal{N}\times \operatorname{BS}_{\Omega}(\Gamma)+\mathcal{N}\times \nabla_{\partial\Omega}g$ in $W^{-\frac{1}{2},2}(\partial\Omega)$ which by the given construction in \Cref{3P18} is a non-trivial element of $\operatorname{Ker}(\operatorname{BS}_{\partial\Omega})$.
	\newline
	\newline
	\underline{(ii):} It follows from the continuity of the Biot-Savart operator, \Cref{CL1}, that $\operatorname{BS}_{\partial\Omega}(j_n)$ converges weakly in $L^2(\Omega)$ to $\operatorname{BS}_{\partial\Omega}(j)=0$ since $j\in \operatorname{Ker}(\operatorname{BS}_{\partial\Omega})$.
	\newline
	\newline
	\underline{(iii):} We claim that $\operatorname{BS}_{\partial\Omega}:W^{-\frac{1}{2},2}\mathcal{V}_0(\partial\Omega)\rightarrow C^k\mathcal{V}(\overline{P})$ is a continuous operator for any fixed $k\in \mathbb{N}_0$. The assertion will then follow from the fact that $C^{k+1}(\overline{P})$ embeds compactly into $C^k(\overline{P})$ and the fact that $\operatorname{BS}_{\partial\Omega}(j_n)$ converges weakly to $0$ in $L^2(\Omega)$. We prove this fact for $k=0$ since the general case can be obtained in the same spirit. For this purpose we assume $j\in L^2\mathcal{V}_0(\partial\Omega)$ and write
	\begin{gather}
		\nonumber
		4\pi\operatorname{BS}_{\partial\Omega}(j)(x)=\epsilon_{imk}\int_{\partial\Omega}j^i(y)\frac{x^m-y^m}{|x-y|^3}d\sigma(y)e_k 
	\end{gather}
	where we use the Einstein summation convention. We define now for fixed $1\leq k\leq 3$ the vector field $\psi^k_x(y):=\epsilon_{lmk}\frac{x^m-y^m}{|x-y|^3}e_l$ which defines a smooth vector field on $\partial\Omega$ which is not-necessarily tangent to $\partial\Omega$. We may further set $\phi^k_x(y):=\psi^k_x(y)-(\mathcal{N}(y)\cdot \psi^k_x(y))\mathcal{N}(y)$ which gives rise to a $C^{0,1}\mathcal{V}(\partial\Omega)$-vector field. We hence conclude
	\begin{gather}
		\nonumber
		4\pi|\operatorname{BS}_{\partial\Omega}(j)(x)|\leq \sum_{k=1}^3\left|\int_{\partial\Omega}j(y)\cdot \phi^k_x(y)d\sigma(y)\right|\leq \|j\|_{W^{-\frac{1}{2},2}(\partial\Omega)}\sum_{k=1}^3\|\phi^k_x\|_{W^{\frac{1}{2},2}(\partial\Omega)}
	\end{gather}
	by definition of the norm $W^{-\frac{1}{2},2}$. It is now easy to see that because $|x-y|\geq \delta$ for all $x\in P$ and $y\in \partial\Omega$ for a suitable constant $\delta>0$ independent of $x$ and $y$ that we can bound $\|\phi^k_x\|_{W^{\frac{1}{2},2}}(\partial\Omega)$ above by some constant independent of $x$. Taking the supremum on both sides of the inequality we find $\|\operatorname{BS}_{\partial\Omega}(j)\|_{C^0(P)}\leq c\|j\|_{W^{-\frac{1}{2},2}(\partial\Omega)}$ for some $c>0$ independent of $j$. As explained before the general case follows in a similar fashion and the result follows by employing compact embeddings.
\end{proof}
\section*{Acknowledgements}
This work has been supported by the Inria AEX StellaCage.
\appendix
\section{Some regularity results}
We define here $\mathcal{H}_N(\Omega):=\{\Gamma\in L^2\mathcal{V}(\Omega)\mid \operatorname{curl}(\Gamma)=0=\operatorname{div}(\Gamma)\text{, }\Gamma\parallel \partial \Omega\}$, where the condition $\operatorname{div}(\Gamma)=0$ and $\Gamma\parallel \partial\Omega$ is understood in the sense that $\Gamma$ is $L^2$-orthogonal to all $\nabla f$, $f\in H^1(\Omega)$.
\begin{lem}
	\label{ALemma1}
	Let $\Omega\subset \mathbb{R}^3$ be a bounded $C^{1,1}$-domain. Then $\mathcal{H}_N(\Omega)\subset \bigcap_{1<p<\infty}W^{1,p}\mathcal{V}(\Omega)$.
\end{lem}
\begin{proof}[Proof of \Cref{ALemma1}]
	It follows first from \cite[Theorem 2.9]{ABDG98} that $\mathcal{H}_N(\Omega)\subset W^{1,2}\mathcal{V}(\Omega)$. Next we can localise the problem by flattening the boundary (the interior case can be handled identically by working in a local chart). We observe that locally around any fixed point $x\in \partial\Omega$, $\operatorname{curl}(\Gamma)=0$ implies that we may write $\Gamma=\operatorname{grad}(f)$ for some $f\in H^1(U)$, where $U=V\cap \overline{\Omega}$ is a suitable open neighbourhood of $x$. We see that $\Delta f=0$ in $U$ and $\mathcal{N}\cdot \nabla f=0$ on $\partial \Omega\cap U$. In order to show that $f$ is locally of class $W^{2,p}$ for every $1<p<\infty$, we compose $f$ with a boundary chart $\mu:U\rightarrow \{(x,y,z)|z\geq 0\}$ and multiply it by a suitable bump function $\rho$ (with an abuse of notation we write, for notational simplicity, $f$ for the composition $f\circ \mu^{-1}$). Letting $\tilde{f}:=\rho f$ we see that $\Delta_g\tilde{f}=f\Delta_g\rho-2g(\nabla_g f, \nabla_g \rho)$ and $\mathcal{N}_g\cdot \operatorname{grad}_g(\tilde{f})=\mathcal{N}^i_g\partial_i\rho f$, where the quantities are now considered with respect to a pulled back metric $g$ with $C^{0,1}$-coefficients $g_{ij}$. Since we multiplied $f$ by a bump function we may further assume that $\tilde{f}$ is defined on some smoothly bounded domain. We recall that we already know that $\tilde{f}\in W^{2,2}$ because $\Gamma\in W^{1,2}\mathcal{V}(\Omega)$, thus \cite[Lemma 2.4.1.4]{Gris85} together with a standard bootstrapping argument implies that $\tilde{f}\in W^{2,p}$ for all $1<p<\infty$ and consequently $f\in W^{2,p}$ after possibly shrinking the neighbourhood $U$ and thus $\Gamma \in W^{1,p}\mathcal{V}(\Omega)$ by compactness of $\overline{\Omega}$.
\end{proof}
We will also need the corresponding result for normal boundary conditions. To this end we let $\mathcal{H}_D(\Omega):=\{\Theta\in L^2\mathcal{V}(\Omega)\mid \operatorname{curl}(\Theta)=0=\operatorname{div}(\Theta)\text{, }\Theta\perp \partial\Omega\}$.
\begin{lem}
	\label{ALemma2}
	Let $\Omega\subset\mathbb{R}^3$ be a bounded $C^{1,1}$-domain. Then $\mathcal{H}_D(\Omega)\subset \bigcap_{1<p<\infty}W^{1,p}\mathcal{V}(\Omega)$.
\end{lem}
\begin{proof}
	It follows first from \cite[Theorem 2.12]{ABDG98} that $\mathcal{H}_D(\Omega)\subset H^1\mathcal{V}(\Omega)$. We can then just like in the proof of \Cref{ALemma1} localise the problem and hence assume that $\Theta=\nabla f$ for a suitable function $f\in W^{2,2}$. We observe that the condition $\nabla f\times \mathcal{N}=0$ implies that $f$ is locally constant along $\partial\Omega$ and so it satisfies appropriate Dirichlet-boundary conditions. We can then similarly employ \cite[Theorem 2.4.1.4]{Gris85} and argue identical in spirit as in the proof of \Cref{ALemma1} to deduce the result.
\end{proof}
\section{An $L^2$-Hodge decomposition theorem}
\begin{thm}
	\label{BTheorem1}
	Let $\Omega\subset \mathbb{R}^3$ be a bounded $C^{1,1}$-domain, then for every $v\in L^2\mathcal{V}(\Omega)$ there exists some $f\in W^{1,2}_0(\Omega)$, $A\in W^{1,2}\mathcal{V}(\Omega)$, $\operatorname{div}(A)=0$, $A\perp \partial\Omega$, $h\in W^{1,2}(\Omega)$ with $\Delta h=0$ in the weak sense in $\Omega$ and $\Gamma\in \mathcal{H}_N(\Omega)$ with
	\begin{gather}
		\nonumber
		v=\operatorname{grad}(f)+\operatorname{curl}(A)+\operatorname{grad}(h)+\Gamma
	\end{gather}
	and this decomposition is $L^2$-orthogonal. Further, $f,h$ and $A$ may be chosen such that
	\begin{gather}
		\nonumber
		\|f\|_{W^{1,2}(\Omega)}+\|h\|_{W^{1,2}(\Omega)}+\|A\|_{W^{1,2}(\Omega)}+\|\Gamma\|_{L^2(\Omega)}\leq c\|v\|_{L^2(\Omega)}
	\end{gather}
	for some $c>0$ independent of $v$.
\end{thm}
\begin{proof}
	We start by considering the space $\mathcal{X}:=\{\nabla \rho\mid \rho \in H^1(\Omega)\}$. Upon subtracting the mean value of a given $\rho\in H^1(\Omega)$ we see that we may restrict attention to functions of zero mean. Then the Poincar\'{e} inequality implies that this is an $L^2$-closed subspace. Similarly we may now consider the subspace $\mathcal{Y}:=\{\nabla f|f\in H^1_0(\Omega)\}\leq \mathcal{X}$. Again an application of the Poincar\'{e} inequality implies the $L^2$-closedness of this subspace. We can therefore express any $v\in L^2\mathcal{V}(\Omega)$ as $v=\operatorname{grad}(f)+\operatorname{grad}(h)+w$ where $w$ is $L^2$-orthogonal to $\mathcal{X}$. We note that $\operatorname{grad}(h)$ being $L^2$-orthogonal to $\mathcal{Y}$ is equivalent to the statement $\Delta h=0$ in the weak sense. We now consider the space $\mathcal{Z}:=\{\operatorname{curl}(A)|A\in W^{1,2}\mathcal{V}(\Omega)\text{, }\operatorname{div}(A)=0\text{, }A\perp \partial\Omega\}$. We can consider the subspace $\mathcal{H}_D(\Omega)=\{\Theta\in L^2\mathcal{V}(\Omega)|\operatorname{curl}(\Theta)=0=\operatorname{div}(\Theta)\text{, }\Theta \perp \partial\Omega\}$ and observe that $\mathcal{H}_D(\Omega)\subset W^{1,2}\mathcal{V}(\Omega)$, c.f. \cite[Theorem 2.12]{ABDG98}. Then the vector potential $A$ of any element of the space $\mathcal{Z}$ may be additionally assumed to be $L^2$-orthogonal to the space $\mathcal{H}_D(\Omega)$. It then follows identically as in \cite[Theorem 2.1.5, Corollary 2.1.6 \& Proposition 2.2.3]{S95}, keeping in mind that the principal curvatures of a $C^{1,1}$-boundary are essentially bounded functions, that we have the Poincar\'{e}-type inequality $\|A\|_{W^{1,2}(\Omega)}\leq c\|\operatorname{curl}(A)\|_{L^2(\Omega)}$ for some constant $c>0$ independent of $A$. From this it follows easily that $\mathcal{Z}$ is $L^2$-closed. We observe that $\mathcal{Z}$ is $L^2$-orthogonal to $\mathcal{X}$ and thus we may decompose further $v=\operatorname{grad}(f)+\operatorname{grad}(h)+\operatorname{curl}(A)+\Gamma$. We claim that $\Gamma\in \mathcal{H}_N(\Omega)$. Since $\Gamma$ is $L^2$-orthogonal to $\mathcal{X}$ we only need to argue that $\Gamma$ is $L^2$-orthogonal to every vector field of the form $\operatorname{curl}(\psi)$ for some smooth vector field $\psi$ which is compactly supported in $\Omega$. We see that $\psi$ is normal to the boundary of $\Omega$ and we can then find a solution $\phi$ to the Dirichlet problem $\Delta \phi=-\operatorname{div}(\psi)$ in $\Omega$ and $\phi|_{\partial\Omega}=0$ of class $W^{2,2}(\Omega)$, \cite[Theorem 2.4.2.5]{Gris85}, so that $A:=\psi+\nabla \phi$ remains normal to the boundary and is divergence-free. Therefore $\operatorname{curl}(\psi)=\operatorname{curl}(A)\in \mathcal{Z}$ and thus $\Gamma$ must be $L^2$-orthogonal to $\operatorname{curl}(\psi)$ implying $\operatorname{curl}(\Gamma)=0$. The a priori estimate follows readily from the arguments provided to establish the $L^2$-closedness of the spaces involved and the $L^2$-orthogonality of this decomposition.
\end{proof}
\section{Tangent traces and Biot-Savart operator}
\label{AppC}
Let $\Omega\subset\mathbb{R}^3$ be a bounded $C^{1,1}$-domain. We introduce the following norm on $W^{\frac{1}{2},2}\mathcal{V}(\partial\Omega)$
\begin{gather}
	\nonumber
	\|j\|:=\sup_{\phi\in W^{\frac{1}{2},2}\mathcal{V}(\partial\Omega)\setminus \{0\}}\frac{\left|\int_{\partial\Omega}j\cdot \phi d\sigma\right|}{\|\phi\|_{W^{\frac{1}{2},2}(\partial\Omega)}}.
\end{gather}
It is then standard, \cite[Theorem 2.11]{GR86}, see also \cite{BCS02}, that the following trace map, also referred to as the twisted tangential trace,
\[
\operatorname{Tr}:\left(H^1\mathcal{V}(\Omega),\|\cdot\|_{L^2,\operatorname{curl}}\right)\rightarrow \left(W^{\frac{1}{2},2}\mathcal{V}(\partial\Omega),\|\cdot\|\right)\text{, }v\mapsto \mathcal{N}\times v
\]
is a continuous linear map. Letting $W^{-\frac{1}{2},2}\mathcal{V}(\partial\Omega)$ denote the completion of $W^{\frac{1}{2},2}\mathcal{V}(\partial\Omega)$ with respect to $\|\cdot\|$ it follows that the trace map $\operatorname{Tr}$ extends to a unique trace map, denoted in the same way, $\operatorname{Tr}:H(\operatorname{curl},\Omega)\rightarrow W^{-\frac{1}{2},2}\mathcal{V}(\partial\Omega)$. We note that for any $j\in L^2\mathcal{V}(\partial\Omega)$ we have $\left|\int_{\partial\Omega}j\cdot \phi d\sigma\right|\leq \|j\|_{L^2}\|\phi\|_{L^2}\leq \|j\|_{L^2}\|\phi\|_{W^{\frac{1}{2},2}}$ for all $\phi\in W^{\frac{1}{2},2}\mathcal{V}(\partial\Omega)$ and therefore $L^2\mathcal{V}(\partial\Omega)\subset W^{-\frac{1}{2},2}\mathcal{V}(\partial\Omega)$. In addition, the notion of an $L^2(\partial\Omega)$-tangent trace for elements $w\in H(\operatorname{curl},\Omega)$, i.e. $\mathcal{N}\times w=j$ with $j\in L^2\mathcal{V}(\partial\Omega)$, introduced at the beginning of section 5.1 coincides with the trace of $w$ in the sense introduced here.

Furthermore, it follows from the proof of \Cref{4EXTRALemma} that the extension operator $T$ is continuous if we equip $L^2\mathcal{V}_0(\partial\Omega)$ with the norm $\|\cdot\|$ which follows from the fact that the functional $J$ introduced in the proof satisfies the estimate $|J(\psi)|\leq \|j\|\|\psi\|_{W^{\frac{1}{2},2}}$ and the standard trace estimate for Sobolev functions. If we now let $W^{-\frac{1}{2},2}\mathcal{V}_0(\partial\Omega)$ denote the completion of $L^2\mathcal{V}_0(\partial\Omega)$ with respect to $\|\cdot\|$, then the operator $T$ defined in \Cref{4EXTRALemma} extends to a continuous linear operator $T:W^{-\frac{1}{2},2}\mathcal{V}_0(\partial\Omega)\rightarrow H(\operatorname{curl},\Omega)$ and by an approximation argument it still satisfies $\operatorname{div}(T(j))=0$ in $\Omega$ and $\operatorname{Tr}(T(j))=j$ for all $j\in W^{-\frac{1}{2},2}\mathcal{V}_0(\partial\Omega)$. In addition, since we now know that the operator $T$ is continuous from $\left(L^2\mathcal{V}_0(\partial\Omega),\|\cdot\|\right)$ into $H(\operatorname{curl},\Omega)$ it follows from \Cref{3L14} that
\begin{gather}
	\nonumber
\operatorname{BS}_{\partial\Omega}:\left(L^2\mathcal{V}_0(\partial\Omega),\|\cdot\|\right)\rightarrow L^2\mathcal{H}(\Omega)\text{, }j\mapsto\frac{1}{4\pi}\int_{\partial\Omega}j(y)\times \frac{x-y}{|x-y|^3}d\sigma(y)
\end{gather}
is continuous and hence extends to a unique continuous, linear operator
\[
\operatorname{BS}_{\partial\Omega}:W^{-\frac{1}{2},2}\mathcal{V}_0(\partial\Omega)\rightarrow L^2\mathcal{H}(\Omega).
\]
We collect these observations in the following lemma.
\begin{lem}
	\label{CL1}
	Let $\Omega\subset\mathbb{R}^3$ be a bounded $C^{1,1}$-domain. Then there are unique continuous, linear operators
	\begin{gather}
		\nonumber
		\operatorname{Tr}:H(\operatorname{curl},\Omega)\rightarrow W^{-\frac{1}{2},2}\mathcal{V}(\partial\Omega),
		\\
		\nonumber
		\operatorname{BS}_{\partial\Omega}:W^{-\frac{1}{2},2}\mathcal{V}_0(\partial\Omega)\rightarrow L^2\mathcal{H}(\Omega) 
	\end{gather}
	such that $\operatorname{Tr}(w)=\mathcal{N}\times w$ for all $w\in H^1\mathcal{V}(\Omega)$ and $\operatorname{BS}_{\partial\Omega}(j)(x)=\frac{1}{4\pi}\int_{\partial\Omega}j(y)\times \frac{x-y}{|x-y|^3}d\sigma(y)$ for all $j\in L^2\mathcal{V}_0(\partial\Omega)$.
\end{lem}
\begin{rem}
	The above considerations together with the proofs provided in section 5.1 in fact imply that the conclusions of \Cref{3T11} remain valid if the domain of $\operatorname{BS}_{\partial\Omega}$ is replaced by $W^{-\frac{1}{2},2}\mathcal{V}_0(\partial\Omega)$.
\end{rem}
\bibliographystyle{plain}
\bibliography{mybibfileNOHYPERLINK}
\footnotesize
\end{document}